\documentclass[10pt]{article}

\usepackage{amsmath}
\usepackage{amsfonts}
\usepackage{amsthm}
\usepackage[english]{babel}
\usepackage{graphicx}
\usepackage[all]{xy}

\newtheorem{theorem}{Theorem}
\newtheorem{lemma}{Lemma}
\newtheorem{definition}{Definition}
\newtheorem{corollary}{Corollary}

\newtheorem{remark}{Remark}

\newcommand{\mR}{\mathbb{R}}
\newcommand{\mC}{\mathbb{C}}
\newcommand{\mN}{\mathbb{N}}
\newcommand{\mE}{\mathbb{E}}

\newcommand{\mS}{\mathbb{S}}

\newcommand{\mQ}{\mathbb{Q}}

\newcommand{\cD}{\mathcal{D}}
\newcommand{\cM}{\mathcal{M}}
\newcommand{\cH}{\mathcal{H}}

\newcommand{\cF}{\mathcal{F}}
\newcommand{\cP}{\mathcal{P}}
\newcommand{\cC}{\mathcal{C}}

\newcommand{\cS}{\mathcal{S}}
\newcommand{\cG}{\mathcal{G}}
\newcommand{\cW}{\mathcal{W}}
\newcommand{\cO}{\mathcal{O}}
\newcommand{\cV}{\mathcal{V}}
\newcommand{\cN}{\mathcal{N}}
\newcommand{\cU}{\mathcal{U}}

\newcommand{\ux}{\underline{x}}

\newcommand{\uxb}{\underline{x} \grave{}}

\newcommand{\uyb}{\underline{y} \grave{}}
\newcommand{\uy}{\underline{y}}

\newcommand{\uv}{\underline{v}}
\newcommand{\ua}{\underline{\alpha}}
\newcommand{\uom}{\underline{\omega}}

\newcommand{\osp}{\mathfrak{osp}(m|2n)}

\begin{document}
\title{The orthosymplectic supergroup in harmonic analysis}

\author{K.\ Coulembier\thanks{Ph.D. Fellow of the Research Foundation - Flanders (FWO), E-mail: {\tt Coulembier@cage.ugent.be}}}

\date{\small{Department of Mathematical Analysis}\\
\small{Faculty of Engineering -- Ghent University\\ Krijgslaan 281, 9000 Gent,
Belgium}\\
\vspace{4mm}
\small{School of Mathematics and Statistics}\\
\small{University of Sydney\\ Sydney, Australia}}

\maketitle

\begin{abstract}
The orthosymplectic supergroup $OSp(m|2n)$ is introduced as the supergroup of isometries of flat Riemannian superspace $\mR^{m|2n}$ which stabilize the origin. It also corresponds to the supergroup of isometries of the supersphere $\mS^{m-1|2n}$. The Laplace operator and norm squared on $\mR^{m|2n}$, which generate $\mathfrak{sl}_2$, are orthosymplectically invariant, therefore we obtain the Howe dual pair $(\mathfrak{osp}(m|2n),\mathfrak{sl}_2)$. This Howe dual pair solves the problems of the dual pair $(SO(m)\times Sp(2n),\mathfrak{sl}_2)$, considered in previous papers. In particular we characterize the invariant functions on flat Riemannian superspace and show that the integration over the supersphere is uniquely defined by its orthosymplectic invariance. The supersphere manifold is also introduced in a mathematically rigorous way. Finally we study the representations of $\mathfrak{osp}(m|2n)$ on spherical harmonics. This corresponds to the decomposition of the supersymmetric tensor space of the $m|2n$-dimensional super vectorspace under the action of $\mathfrak{sl}_2\times\osp$. As a side result we obtain information about the irreducible $\osp$-representations $L_{(k,0,\cdots,0)}^{m|2n}$. In particular we find branching rules with respect to $\mathfrak{osp}(m-1|2n)$.
\end{abstract}

\textbf{MSC 2000 :}   17B10, 58C50, 17B15\\
\noindent
\textbf{Keywords :} Howe dual pair, orthosymplectic supergroup and superalgebra, not completely reducible representations, invariant integration on supermanifolds, supersphere, 

\section{Introduction}

In recent work, we have been developing a new approach to the study of supergeometry, by means of harmonic analysis, see e.g. \cite{MR2539324, CDBS3, Mehler, DBS9, DBE1} and references therein. We consider a superspace $\mR^{m|2n}$ generated by $m$ commuting or bosonic variables and $2n$ anti-commuting or fermionic variables. The main feature of this approach is the introduction of a super Laplace operator $\nabla^2$ and a super norm squared $R^2$. These generate the Lie algebra $\mathfrak{sl}_2$ and are invariant under the action of the product of the orthogonal and the symplectic group $O(m)\times Sp(2n)$. This leads to the pair $\left(SO(m)\times Sp(2n),\mathfrak{sl}_2\right)$ as a generalization of the Howe dual pair $(SO(m),\mathfrak{sl}_2)$ for harmonic analysis on $\mR^m$, see \cite{MR0986027}. However this pair is not complete. The consequences of this incompleteness are listed at the beginning of section \ref{introOSp}. In short, there are too many invariant functions and linear maps on $\mR^{m|2n}$ and the spaces of spherical harmonics (i.e. polynomial null-solutions of $\nabla^2$ of a fixed degree) are not irreducible $SO(m)\times Sp(2n)$-representations. This implies that the polynomials do not have a multiplicity free decomposition into irreducible pieces under the joint action of the dual pair, which is a fundamental property of Howe dual pairs, see \cite{MR2028498, MR0986027, MR1272070, MR1827078}. A similar situation occurs in the theory of Dunkl operators, with pair $(G,\mathfrak{sl}_2)$ where $G <O(m)$ is a Coxeter group, see e.g. \cite{MR1827871}.

These problems will be solved by considering the orthosymplectic superalgebra $\mathfrak{osp}(m|2n)$, which has not been studied previously in our approach to super harmonic analysis. We introduce Lie supergroups in the sheaf-theoretical approach of Berezin and Leites, see e.g. \cite{Berezin, MR0998351, Leites}. This category is equivalent with the category of super Harish-Chandra pairs, see e.g. \cite{MR0580292, Vishnyakova}. We explicitly prove that the notions of invariance of functions and linear maps with respect to Lie supergroups, and irreducibility of corepresentations of Lie supergroups can be characterized in the expected way by the Harish-Chandra pair.

In classical harmonic analysis, the orthogonal group $O(m)$ corresponds to the group of isometries of $\mR^m$ which stabilize the origin. We use the definition in \cite{MR2434470} of the supergroup of isometries of a Riemannian supermanifold. By doing so we obtain the Lie supergroup $OSp(m|2n)$, corresponding to the Harish-Chandra pair $(O(m)\times Sp(2n), \mathfrak{osp}(m|2n))$ as the isometry supergroup to consider on $\mR^{m|2n}$. Since the generators of $\mathfrak{sl}_2$ are $OSp(m|2n)$-invariant we obtain the dual pair $(\mathfrak{osp}(m|2n),\mathfrak{sl}_2)$ or $(SOSp(m|2n),\mathfrak{sl}_2)$. Because this dual pair solves the aforementioned problems it is the correct Howe dual pair for super harmonic analysis on $\mR^{m|2n}$. The bigger algebra into which the dual pair is embedded is $\mathfrak{osp}(4n+1|2m)\supset \osp\times \mathfrak{sl}_2$. A more general theory of Howe dual pairs with orthosymplectic algebras can be found in \cite{MR2028498, MR1272070, MR1827078}. 

We prove that the integration over the supersphere as introduced in \cite{DBE1, MR2539324} is orthosymplectically invariant, which yields a nice unique characterization in Theorem \ref{Pizuniekosp}. The explicit expression for an invariant integration over the supersphere is important for the generalization of field theories on the sphere to superspace, see e.g. \cite{MR1897209, MR2125586} and for the theory of invariant integration on supergroups, see e.g. \cite{MR2667819, MR1845224, MR2172158}. In the approach of the present paper the supersphere manifold $\mS^{m-1|2n}$ is constructed in a mathematically rigorous way and imbedded in flat superspace. The supersphere integral can then also be seen as the unique orthosymplectically invariant integral on $\mS^{m-1|2n}$. The existence and unicity of this integral can also be deduced from the theory in \cite{MR2667819}. 

In \cite{MR2395482} it was proven that the spaces of spherical harmonics $\cH_k$ of fixed degree are irreducible $\mathfrak{osp}(m|2n)$-modules if $m-2n\ge 2$. In Theorem \ref{irrH} we prove that this result still holds for all $(m,n)$ such that $m-2n\not\in-2\mN$. This gives a multiplicity free irreducible direct sum decomposition for the supersymmetric tensor space $T(V)$ under the action of $\mathfrak{sl}_2\times \osp$, with $V=L^{m|2n}_{(1,0,\cdots,0)}$ the natural representation space for $\osp$. This implies that $(\osp,\mathfrak{sl}_2)$ is a Howe dual pair for $T(V)\cong \cP$ with $\cP$ the polynomials on $\mR^{m|2n}$ if $m-2n\not\in-2\mN$.

Also the $\mathfrak{osp}(m|2n)$-representations of spherical harmonics for the case $m-2n\in-2\mN$ are studied. They are always indecomposable, but the irreducibility depends on the degree $k$ of the spherical harmonics $\cH_k$. Using all these results we obtain polynomial realizations of the simple $\osp$-module $L^{m|2n}_{(k,0,\cdots,0)}$ (\cite{MR051963}) with highest weight $(k,0,\cdots,0)$ for all values of $(m,n,k)\in\mN^3$. This generalizes the construction of $L^{m|2n}_{(k,0,\cdots,0)}$ as the tracefree supersymmetric part of tensor products of $L^{m|2n}_{(1,0,\cdots,0)}$ of \cite{MR0621253} to the case were $\cH_k$ is not irreducible in Theorem \ref{Lkrep}. As a side result this gives the dimension of the irreducible representation $L^{m|2n}_{(k,0,\cdots,0)}$ and the possibility to obtain the character formula, which are non trivial since the representations are almost all atypical, see e.g. \cite{MR051963}. We also obtain results for branching of $L^{m|2n}_{(k,0,\cdots,0)}$ as an $\mathfrak{osp}(m-1|2n)$-module, for all distinct cases in Theorem \ref{branchingThm}.

The paper is organized as follows. First a short introduction to supermanifolds, Lie supergroups and harmonic analysis on superspace is given. The supergroup of isometries of $\mR^{m|2n}$ which stabilize the origin is calculated to be $OSp(m|2n)$. Next, we show that the introduction of the Howe dual pair $(\mathfrak{osp}(m|2n),\mathfrak{sl}_2)$ solves the three problems of the dual pair $(SO(m)\times Sp(2n),\mathfrak{sl}_2)$. A transparent and mathematically rigorous definition of the supersphere manifold is given. Subsequently the unicity of the supersphere integral as an orthosymplectically invariant functional is proven for functions on $\mR^{m|2n}$ or on $\mS^{m-1|2n}$. Finally, the $(S)OSp(m|2n)$-corepresentations of spherical harmonics and the representations $L^{m|2n}_{(k,0,\cdots,0)}$ are studied and it is shown that the full space of polynomials on $\mR^{m|2n}$ is an irreducible $\mathfrak{osp}(4n+1|2m)$-module.

\section{Preliminaries}

\label{preliminaries}

\subsection{Supermanifolds and Lie supergroups}
We consider the definition of a supermanifold as in the approach of Berezin, see a.o. \cite{Berezin, MR0580292, Leites}. A sheaf $\cO$ of superalgebras with unity on an $m$-dimensional manifold $\cM_0$, maps each open subset $U$ of $\cM_0$ into a superalgebra $\cO(U)$. In particular, this defines a sheaf of superalgebras on each open subset $V$ of $\cM_0$, $\cO^{V}$ by $\cO^{V}(U)=\cO(U)$ for each open subset $U$ of $V$. We denote the sheaf of smooth functions on a manifold $\cM_0$ by $\cC^\infty_{\cM_0}$ and also use the notation $\cC^\infty(\cM_0)=\cC^\infty_{\cM_0}(\cM_0)$. For a manifold morphism $f:\cM_0\to \cN_0$ and a sheaf $\cO$ on $\cM_0$ the push forward of $\cO$ is a sheaf on $\cN_0$ defined by $f_{\ast}\cO (U)=\cO (f^{-1}(U))$, for each open $U\subset\cN_0$.

\begin{definition}
\label{supmandef}
A supermanifold of super dimension $D|N$ is a ringed space $\cM=(\cM_0,\cO_{\cM})$, with $\cM_0$ the underlying smooth $D$-dimensional manifold and $\cO_{\cM}$ the structure sheaf which is a sheaf of $\mR$-superalgebras with unity on $\cM_0$. The sheaf $\cO_{\cM}$ satisfies the local triviality condition: there exists an open cover $\{U_i\}_{i\in I}$ of $\cM_0$ and isomorphisms of sheaves of superalgebras $T_i:\cO_{\cM}^{U_i}\to \cC^\infty_{U_i}\otimes\Lambda_{N}$, where $\Lambda_{N}$ is the Grassmann algebra generated by $N$ anti-commuting variables.
\end{definition}

The sections of the structure sheaf (the elements of the superalgebra $\cO_{\cM}(\cM_0)$) are referred to as superfunctions on $\cM$. We also use the notation $\cO(\cM)=\cO_{\cM}(\cM_0)$. A morphism of supermanifolds $\Phi:\cM\to\cN$ is a morphism of ringed spaces $(\phi,\phi^\sharp)$. The mapping $\phi$ is a manifold morphism $\cM_0\to\cN_0$ and $\phi^\sharp:\cO_{\cN}\to\phi_{\ast}\cO_{\cM}$ is a morphism of sheaves on $\cN$. The morphism $\Phi$ is entirely determined by its induced mapping of sections $\phi^\sharp_{\cN_0}:\cO(\cN)=\cO_\cN(\cN_0)\to\cO_{\cM}\left(\phi^{-1}(\cN_0)\right)$. This was proven in \cite{MR0580292}. In order to keep the notations simpler we will also denote $\phi^\sharp_{\cN_0}$ by $\phi^\sharp$.

Each point $p\in\cM_0$ defines a morphism $\delta_p:(\{\ast\},\mR)\to(\cM_0,\cO_{\cM}$), with $\{\ast\}$ the manifold consisting of one point. This is defined by 
\begin{eqnarray}
\label{evalpunt}
\delta_p^{\sharp}(f)&=&[f]_0(p)
\end{eqnarray}
where $[\cdot]_0$ is the canonical projection $\cO_{\cM}(\cM_0)\to\cC^\infty_{\cM_0}(\cM_0)$, also denoted by $\delta^\sharp_\cM$. The identity morphism on a supermanifold $\cM$ is given by 
\begin{eqnarray*}
id_\cM=(id_{\cM_0},id^\sharp_{\cM}):&&\cM\to\cM.
\end{eqnarray*}
We define the supermanifold morphisms $\rho:\cM\to\cM\otimes \cM$ given by $\rho_0(u)=u\otimes u$ and $\rho^\sharp(f\otimes g)=fg$ and $C:\cM\to(\ast,\mR)$ defined by $C^\sharp(\lambda)=\lambda 1_{\cM}$ with $1_{\cM}=1_{\cM_0}$ the unit function on $\cM_0$.

\begin{definition}
\label{defLiegr}
A Lie supergroup is a supermanifold $\cG=(\cG_0,\cO_\cG)$ equipped with the additional structure of a supermanifold morphism $\mu=(\mu_0,\mu^\sharp):\cG\otimes\cG\to\cG$, an involutive diffeomorphism $\nu=(\nu_0,\nu^\sharp):\cG\to\cG$ and a distinguished point $e_\cG\in\cG_0$. These satisfy
\begin{eqnarray}
\nonumber
\mu\circ\left(id_{\cG}\times \mu\circ \left(id_\cG\times id_\cG\right)\right)&=&\mu\circ\left(\mu\circ\left( id_\cG\times id_\cG\right)\times id_{\cG}\right)\\
\label{multid}
\mu\circ\left(id_\cG\times \delta_{e_\cG}\right)&=&id_\cG=\mu\circ\left(\delta_{e_\cG}\times id_\cG\right)\\
\nonumber
\mu\circ(id_\cG\times \nu)\circ\rho&=&\delta_{e_{\cG}}\circ C=\mu\circ( \nu\times id_\cG)\circ\rho.
\end{eqnarray}
\end{definition}
With this definition, the underlying manifold $\cG_0$ of a Lie supergroup $\cG$ becomes a Lie group with multiplication $\mu_0$ and inversion $\nu_0$.

The structure sheaf of a Lie supergroup is always globally split, see e.g. \cite{MR0998351, MR0760837}. This means that for a general Lie supergroup $\cG$ of dimension $D|N$, the structure sheaf satisfies
\begin{eqnarray}
\label{sheafsplit}
\cO_{\cG}&\cong &\cC^\infty_{\cG_0}\otimes \Lambda_N.
\end{eqnarray}
In other words, the anti-commuting variables correspond to global coordinates.

For each superalgebra we use the notation $|\cdot|$ for the gradation, $|a|$ is $0$ if $a$ is even and $1$ if $a$ is odd.

The action of a Lie supergroup $\cG=(\cG_0,\cO_{\cG})$ on a supermanifold $\cM=(\cM_0,\cO_{\cM})$ is given by a supermanifold morphism $\Psi:\cG\otimes \cM\to\cM$ which satisfies the conditions
\begin{eqnarray}
\label{multactie}
\Psi\circ\left(id_\cG\otimes \Psi\right)&=&\Psi\circ\left(\mu\otimes id_{\cM}\right),\\
\nonumber
\Psi\circ\left(\delta_{e_\cG}\otimes id_\cM\right)&=&id_\cM.
\end{eqnarray}

The Lie superalgebra $\mathfrak{g}$ corresponding to a Lie supergroup $\cG$ is given by the elements of $\mbox{Der} \cO(G)$ which are right invariant, see \cite{MR2069561}. Here $\mbox{Der} \cO(G)$ is the superalgebra of superderivations of the algebra $\cO(\cG)$.
\begin{definition}
\label{defLiealg1}
The Lie superalgebra $\mathfrak{g}$ corresponding to the Lie supergroup $\cG$ is the algebra with elements given by $\{X\in Der\cO({\cG})| \mu^\sharp \circ X =(X\otimes id_\cG^\sharp)\circ\mu^\sharp\}$ and with superbracket given by the supercommutator $[X,Y]=X\circ Y-(-1)^{|X||Y|}Y\circ X$.
\end{definition}
This invariance property immediately extends to the universal enveloping algebra:
\begin{eqnarray}
\label{invuniv}
\mu^\sharp \circ X &=&\left(X\times id_{\cG}^\sharp\right)\circ\mu^\sharp\qquad\forall X\in\cU(\mathfrak{g}).
\end{eqnarray}

The distinguished point $e_{\cG}$ of $\cG_0$ will also be denoted by $0$. The evaluation in that point \eqref{evalpunt} therefore is denoted by $\delta_0^\sharp$. The following lemma is an immediate consequence of definition \ref{defLiealg1}.
\begin{lemma}
\label{defLiealg2}
An $\mR$-vector space basis for the Lie superalgebra $\mathfrak{g}$ is given by the derivatives $Y_j=\left(\delta^\sharp_{0}\circ \cD_j\otimes id_{\cG}^\sharp\right)\circ\mu^\sharp$ for $\{\cD_j,j=1,\cdots,D\}$ the bosonic derivatives with respect to a set of local coordinates on $\cG_0$ in a neighborhood around the origin and $\{\cD_j,j=D+1,\cdots,D+N\}$ the $N$ global Grassmann derivatives.
\end{lemma}

By means of definition \ref{defLiealg1} we can define a mapping from Lie supergroups to pairs of Lie groups and Lie superalgebras, $\cG\to(\cG_0,\mathfrak{g})$. The Lie group and Lie superalgebra satisfy the compatibility conditions of a super Harish-Chandra pair (also known as a supergroup pair):

\begin{definition}
\label{HaCh}
A Harish-Chandra pair is a pair $(\cG_0,\mathfrak{g})$, consisting of a  Lie group $\cG_0$ and a Lie superalgebra $\mathfrak{g}=\mathfrak{g}_0\oplus \mathfrak{g}_1$ with $\mathfrak{g}_0$ the Lie algebra of $\cG_0$ and a representation Ad of $\cG_0$ on $\mathfrak{g}$ such that
\begin{itemize}
\item Ad on $\mathfrak{g}_0$ is the usual adjoint action,
\item the differential of the action in the identity is equal to the Lie superbracket, restricted to $\mathfrak{g}_0\times\mathfrak{g}$.
\end{itemize}
\end{definition}

The mapping $\cG\to(\cG_0,\mathfrak{g})$ from Lie supergroups to Harish-Chandra pairs described above is an equivalence of categories, see \cite{MR0580292}. The adjoint representation of $\cG_0$ on $\mathfrak{g}$ (in order to complete the Harish-Chandra pair) is defined as 
\begin{eqnarray}
\label{Adj}
Ad(g)X&=&(\delta_g^\sharp\times id_\cG^\sharp)\circ\mu^\sharp\circ X\circ (\delta_{g^{-1}}^\sharp\times id_{\cG}^\sharp)\circ\mu^\sharp.
\end{eqnarray} 
This action of $\cG_0$ on $\mathfrak{g}$ is trivially extended to the universal enveloping algebra $\cU(\mathfrak{g})$. 

As in the classical case, the invariance of a function is defined as follows.
\begin{definition}
\label{definvf}
A superfunction $f\in\cO({\cM})$ is invariant for the action $\Psi$ of a Lie supergroup $\cG$ on the supermanifold $\cM$ if
\begin{eqnarray*}
\psi^{\sharp}(f)&=&1_{\cG}\times f,
\end{eqnarray*}
with $1_\cG\in\cO_{\cG}(\cG_0)$ the unit function on $\cG_0$.
\end{definition}

The definition of invariance of a linear map is given by
\begin{definition}
\label{definvmap}
Consider a supermanifold $\cM=(\cM_0,\cO_\cM)$ and a Lie supergroup $\cG=(\cG_0,\cO_\cG)$ with action $\Psi$ on $\cM$. A linear map $T: \cO(\cM)\to \cV$ for some vector space $\cV$ is invariant with respect to the action $\Psi$ if the equality
\begin{eqnarray*}
\left(id^\sharp_\cG\times T\right)\circ \psi^\sharp&=&1_{\cG}\times T
\end{eqnarray*}
holds as maps $\cO({\cM})\to \cO(\cG)\times \cV$.
\end{definition}

The definition of a corepresentation of a Lie group (corepresentation of the Hopf algebra of functions on the group) trivially extends to the super case.
\begin{definition}
\label{defrep1}
A corepresentation of a Lie supergroup $\cG$ on a graded vector space $\cV$ is an even linear map $\chi:\cV\to\cO(\cG)\times\cV$ satisfying $(\delta^\sharp_{e_\cG}\times id_{\cV})\circ\chi=id_{\cV}$ and $(\tau\times id_{\cV})\circ(id_{\cG}^\sharp\times\chi)\circ\chi=(\mu^\sharp\times id_{\cV})\circ\chi$ with $id_{\cV}$ the identity map on $\cV$, $\mu$ the multiplication of $\cG$ and $\tau$ the graded flip operator, $\tau(f\times g)=(-1)^{|f||g|}g\times f$.
\end{definition}

A corepresentation $\psi^\sharp$ is irreducible if the vector space $\cV$ does not have a true sub-vector space $\cW$ such that $\psi^\sharp$ restricted to $\cW$ forms a corepresentation. A representation of a super Harish-Chandra pair is defined as follows, see e.g. \cite{MR2207328}.
\begin{definition}
\label{defrep2}
A representation of a super Harish-Chandra pair $(\cG_0,\mathfrak{g})$ on a graded vector space $\cV$ is a pair $\Pi=(\pi_0,\rho^\pi)$,
\begin{itemize}
\item $\pi_0$: $\cG_0\to Aut(\cV_0)\times Aut(\cV_1)$ a group morphism
\item $\rho^\pi$: $\mathfrak{g}\to End(\cV)$ a Lie superalgebra morphism
\end{itemize}
such that $d\pi_0=\rho^\pi$ on $\mathfrak{g_0}$ and $\rho^\pi\left(Ad(g)X\right)=Ad(\pi_0(g))\rho^\pi(X)$ for $X\in\mathfrak{g}$ and $g\in\cG_0$.
\end{definition}
A representation $\Pi$ is irreducible if the vector space $\cV$ does not have a sub-vector space $\cW$ such that $\Pi$ restricted to $\cW$ forms a representation.

Definition \ref{defrep1} of a corepresentation $\chi$ of a Lie supergroup $\cG$ on a graded vector space $\cV$ induces a representation of the Harish-Chandra pair of $\cG$. Define
\begin{equation}
\label{indHCrep}
\begin{cases}
&\pi_0(g)=\left(\delta_{g}^\sharp\times id_{\cV}\right)\circ\chi\quad \mbox{for} \quad g\in\cG_0\\
&\rho^\pi(X)=\left(\delta_{e_{\cG}}^{\sharp}\circ X\times id_{\cV}\right)\circ \chi\quad \mbox{for} \quad X\in\mathfrak{g}.
\end{cases}
\end{equation}

\subsection{Harmonic analysis on Euclidean space}
\label{classHarm}
We consider the Euclidean space $\mR^m$ with $m$ variables $\ux=(x_1,\cdots,x_m)$. The standard orthogonal metric leads to the differential operators $\nabla^2_b=\sum_{j=1}^m\partial_{x_j}^2$, $r^2=\sum_{j=1}^mx_j^2$ and $\mE_b=\sum_{j=1}^mx_j\partial_{x_j}$ acting on functions on $\mR^m$. The operators $-\nabla^2_b/2$, $r^2/2$ and $\mE_b+m/2$ generate the Lie algebra $\mathfrak{sl}_2$. Since they are invariant under the action of the orthogonal group $O(m)$ we obtain the Howe dual pair $(SO(m),\mathfrak{sl}_2)$, see \cite{MR0986027}. This duality is captured in the Fischer decomposition of the space of polynomials
\begin{eqnarray*}
\mR[x_1,\cdots,x_m]&=&\bigoplus_{j=0}^\infty\bigoplus_{k=0}^\infty r^{2j}\cH_k^b
\end{eqnarray*}
with $\cH_k^b$ the polynomial null-solutions of $\nabla_b^2$ of degree $k$. The blocks $r^{2j}\cH_k^b$ are exactly the irreducible blocks of the representation of $SO(m)$ on $\mR[x_1,\cdots,x_m]$. The spaces $\mR[r^2]\cH_k^b=\bigoplus_{j=0}^\infty r^{2j}\cH_k^b$ are irreducible lowest weight modules for $\mathfrak{sl}_2$, with lowest weight $k+m/2$ and weight vectors $r^{2j}\cH_k^b$. This implies that $\mR[x_1,\cdots,x_m]$ has a multiplicity free irreducible direct sum decomposition under the action of $\mathfrak{sl}_2\times \mathfrak{o}(m)$.

The embedding of the unit sphere is given by
\begin{eqnarray*}
\pi_{\mS^{m-1}}:& &\mS^{m-1}\to\mR^m.
\end{eqnarray*}
This defines a surjective mapping from smooth functions on $\mR^m$ to smooth functions on $\mS^{m-1}$ by restriction, $\pi_{\mS^{m-1}}^\sharp:\cC^\infty_{\mR^m}\to\cC_{\mS^{m-1}}^\infty$. This evaluation can be symbolically denoted by
\begin{eqnarray*}
\left[\pi_{\mS^{m-1}}^{\sharp}f\right](\uom)=f(\uom)=f\left(\frac{\ux}{r}\right)=\left[f(\ux)\right]_{r=1},
\end{eqnarray*}
where $\uom$ is a unit vector in $\mR^m$, used as a symbolic notation for local coordinates on the unit sphere. The integral $\int_{\mS^{m-1}}$, where $\int_{\mS^{m-1}(u)}$ represents integration over the sphere with radius $u$, is the unique linear map
\begin{eqnarray*}
\int_{\mS^{m-1}}:& &\cC^\infty(\mR^m)\to\cC^\infty(\mR^+);\qquad f\to \int_{\mS^{m-1}(u)}f
\end{eqnarray*}
which is $SO(m)$-invariant and satisfies $\int_{\mS^{m-1}(u)}1=\sigma_mu^{m-1}=\frac{2\pi^{m/2}}{\Gamma(m/2)}u^{m-1}$ and 
\begin{eqnarray*}
\left[f\right]_{r=u_0}=0&\Rightarrow&\int_{\mS^{m-1}(u_0)}f=0.
\end{eqnarray*}
Note that we omit the measure $d\sigma$ on the sphere. The reason is that we see the unit sphere integration as an invariant functional on the space of functions rather that integration over a manifold with invariant measure.

\subsection{Harmonic analysis on superspace}

We repeat some results on the theory of harmonic analysis on the flat supermanifold $\mR^{m|2n}=(\mR^m,\cC^\infty\otimes\Lambda_{2n})$, as developed in \cite{DBS9, DBE1, CDBS3}. The supervector $\bold{x}$ is defined as

\[
\bold{x}=(X_1,\cdots,X_{m+2n})=(\ux,\uxb)=(x_1,\cdots,x_m,{x\grave{}}_1,\cdots,{x\grave{}}_{2n}).
\]

The commutation relations for the Grassmann algebra and the bosonic variables are captured in the relation $X_iX_j=(-1)^{[i][j]}X_jX_i$ with $[i]=0$ if $i\le m$ and $[i]=1$ otherwise. The superdimension is defined as $M=m-2n$. The orthosymplectic metric $g$ on $\mR^{m|2n}$ is defined as $g\in\mR^{(m+2n)\times(m+2n)}$
\begin{eqnarray}
\label{defg}
g&=&\left( \begin{array}{c|c} I_m&0\\ \hline \vspace{-3.5mm} \\0&J
\end{array}
 \right)
\end{eqnarray}
with $J\in\mR^{2n\times 2n}$ given by 
\begin{eqnarray}
\label{Jmatrix}
J&=&\frac{1}{2}\left( \begin{array}{cccccc} 0&-1&&&\\1&0&&&\\&&\ddots&&\\&&&0&-1\\&&&1&0 
\end{array}
 \right).
 \end{eqnarray}

Define $X^j=\sum_iX_ig^{ij}$. The square of the `radial coordinate' is given by
\begin{eqnarray*}
R^2=\langle \bold{x},\bold{x}\rangle=\sum_{j=1}^{m+2n}X^jX_j =\sum_{i=1}^mx_i^2-\sum_{j=1}^n{x\grave{}}_{2j-1}{x\grave{}}_{2j}=r^2+\theta^2,
\end{eqnarray*}
see \cite{CDBS3}. The super Laplace operator is given by $\nabla^2=\nabla^2_b-4\sum_{j=1}^n\partial_{{x\grave{}}_{2j-1}}\partial_{{x\grave{}}_{2j}}$. The super Euler operator is defined as $\mE=\sum_{i=1}^mx_i\partial_{x_i}+\sum_{j=1}^{2n}{x\grave{}}_j\partial_{{x\grave{}}_j}=\mE_b+\mE_f$. The operators $-\nabla^2/2$, $R^2/2$ and $\mE+M/2$, with $M=m-2n$, again generate the Lie algebra $\mathfrak{sl}_2$, see \cite{DBE1}:
\begin{eqnarray}
\nonumber
\left[\nabla^2/2,R^2/2\right]&=&\mE+M/2,\\
\label{sl2rel}
\left[\nabla^2/2,\mE+M/2\right]&=&2\nabla^2/2\quad \mbox{and}\\
\nonumber
\left[R^2/2,\mE+M/2\right]&=&-2R^2/2.
\end{eqnarray}

The space of super polynomials is given by 
\begin{eqnarray*}
\cP&=&\mR[x_1,\cdots,x_m]\otimes \Lambda_{2n}\subset \cO(\mR^{m|2n})=\cC^\infty(\mR^m)\otimes\Lambda_{2n}.
\end{eqnarray*}
The polynomials of degree $k$ are the elements $P\in\cP$ which satisfy $\mE P=kP$. The corresponding space is denoted by $\cP_k$. The null-solutions of the super Laplace operator are called harmonic superfunctions. The space of the spherical harmonics of degree $k$ is denoted by $\cH_k=\cP_k\cap\ker \nabla^2$. In the purely fermionic case we use the notation $\cH_{k}^{f}$. The Fischer decomposition holds in superspace when the superdimension is not even and negative or in the purely fermionic case, see \cite{DBE1}.

\begin{lemma}[Fischer decomposition]
If $M=m-2n \not \in -2 \mN$, $\cP$ decomposes as
\begin{eqnarray}
\cP = \bigoplus_{k=0}^{\infty} \cP_k= \bigoplus_{j=0}^{\infty} \bigoplus_{k=0}^{\infty} R^{2j}\cH_k.
\label{superFischer}
\end{eqnarray}
In case $m=0$, the decomposition is given by $\Lambda_{2n} = \bigoplus_{k=0}^{n} \left(\bigoplus_{j=0}^{n-k} \theta^{2j} \cH^f_k \right)$.
\label{superFischerLemma}
\end{lemma}

The dimension of $\cH_k$, if $m\not=0$, is given by
\begin{eqnarray}
\label{dimHk}
\dim\cH_k&=&\sum_{i=0}^{\min(k,2n)}\binom{2n}{i}\binom{k-i+m-1}{m-1}-\sum_{i=0}^{\min(k-2,2n)}\binom{2n}{i}\binom{k-i+m-3}{m-1},
\end{eqnarray}
see \cite{DBE1}.

Superfunctions are often defined as finite Taylor expansions.
\begin{definition}
\label{fTaylor}
Consider a function $f=f_A{x\grave{}}_A\in\cC^\infty(\mR^d)\otimes\Lambda_{2p}$ and even superfunctions $\alpha_j(\bold{x})\in\cC^\infty(U)\otimes[\Lambda_{2n}]_0$, $j=1,\cdots,d$ and odd superfunctions $\beta_k(\bold{x})\in\cC^\infty(U)\otimes[\Lambda_{2n}]_1$, $k=1,\cdots,2p$ with $U$ an open subset of $\mR^m$. We expand the $\alpha_j$ as $\alpha_j(\bold{x})=[\alpha_j]_0(\ux)+\widetilde{\alpha}_j(\bold{x})$. In this expansion, $[\alpha_j]_0(\ux)$ is the bosonic part and $\widetilde{\alpha}_j(\bold{x})$ is nilpotent (in particular $\widetilde{\alpha}_j(\bold{x})^{n+1}=0$). The superfunction 
\[f(\alpha_1(\bold{x}),\cdots,\alpha_d(\bold{x}),\beta_1(\bold{x}),\cdots,\beta_{2p}(\bold{x}))\in\cC^\infty(U)\otimes\Lambda_{2n}\]
is defined as
\begin{eqnarray*}
\sum_{A}\left[\sum_{i_1,\cdots,i_d=0}^{n}\frac{f_A^{(i_1,\cdots,i_d)}\left([\alpha_1]_0(\ux),\cdots,[\alpha_d]_0(\ux)\right)}{i_1!\cdots i_d!}\widetilde{\alpha}_1(\bold{x})^{i_1}\cdots\widetilde{\alpha}_d(\bold{x})^{i_d}\right]\beta_A(\bold{x}),
\end{eqnarray*}
where $\beta_A(\bold{x})$ is defined as $\beta_{1}(\bold{x})^{a_1}\cdots\beta_{2p}(\bold{x})^{a_{2p}}$.
\end{definition}

The inner product of two supervectors $\bold{x}$ and $\bold{y}$ is given by
\begin{eqnarray}
\label{inprod}
\langle \bold{x},\bold{y}\rangle=\sum_{i,j=1}^{m+2n}X_ig^{ij}Y_j&=&\sum_{i=1}^mx_iy_i-\frac{1}{2}\sum_{j=1}^n({x\grave{}}_{2j-1}{y\grave{}}_{2j}-{x\grave{}}_{2j}{y\grave{}}_{2j-1}).
\end{eqnarray}
The commutation relations for two supervectors are determined by the relation $X_iY_j=$ $(-1)^{[i][j]}Y_jX_i$. This implies that the inner product \eqref{inprod} is symmetric, i.e. $\langle \bold{x},\bold{y}\rangle=\langle\bold{y},\bold{x}\rangle$.

The orthosymplectic superalgebra $\mathfrak{osp}(m|2n)$ can be generated by the following differential operators on $\mR^{m|2n}$ (see \cite{MR2395482})
\begin{eqnarray}
\label{ospgen}
L_{ij}&=&X_i\partial_{X^j}-(-1)^{[i][j]}X_j\partial_{X^i}
\end{eqnarray}
for $1\le i \le j \le m+2n$. The Laplace-Beltrami operator is defined as
\begin{eqnarray}
\label{LB}
\Delta_{LB}&=&R^2\nabla^2-\mE(M-2+\mE).
\end{eqnarray}
The Laplace-Beltrami operator is a Casimir operator of $\mathfrak{osp}(m|2n)$ of degree $2$ and can be expressed as (see \cite{MR0546778})
\begin{eqnarray}
\label{LBosp}
\Delta_{LB}&=&-\frac{1}{2}\sum_{i,j,k,l=1}^{m+2n}L_{ij}g^{il}g^{jk}L_{kl}.
\end{eqnarray} 
Up to an additive constant this also corresponds to a Casimir operator for $\mathfrak{sl}_2$.

The matrix Lie group $O(m)\times Sp(2n)$ corresponds to all the matrices $S\in\mR^{(m+2n)\times(m+2n)}$ satisfying
\begin{eqnarray}
\label{OxSpdef}
\langle S\cdot\bold{x},S\cdot\bold{y}\rangle=\langle\bold{x},\bold{y}\rangle,
\end{eqnarray}
which is equivalent with $S^TgS=g$ or $S\cdot R^2=R^2$. Each such matrix $S$ is a block matrix $S=\left( \begin{array}{c|c} A&0\\ \hline \vspace{-3.5mm} \\0&B
\end{array}
 \right)$ with $A\in\mR^{m\times m}$ satisfying $A^TA=I_m$ and $B\in\mR^{2n\times 2n}$ satisfying $B^TJB=J$. The matrix $J$ is given in equation \eqref{Jmatrix}. The action of $O(m)\times Sp(2n)$ on $\cC^\infty(\mR^m)\times\Lambda_{2n}$ is given by
\begin{eqnarray}
\label{actieOxSp}
\left(S,f(\bold{x})\right)\to f(S^{-1}\cdot\bold{x})
\end{eqnarray}
such that $(S\cdot\bold{x})_j=\sum_{j=1}^{m+2n}S_{jk}X_k$. The action of the Lie algebra $\mathfrak{o}(m)\oplus\mathfrak{sp}(2n)$ on $\mR^{m|2n}$ in equation \eqref{ospgen} is the differential in the origin of the action of the Lie group $O(m)\times Sp(2n)$.

Since $R^2$, $\nabla^2$ and $\mE+M/2$ generate $\mathfrak{sl}_2$ and are $O(m)\times Sp(2n)$-invariant, we obtain the dual pair $\left(SO(m)\times Sp(2n),\mathfrak{sl}_2\right)$. This is again closely related to the Fischer decomposition \eqref{superFischer} for $M\not\in-2\mN$ or $m=0$. The blocks $\bigoplus_{j}R^{2j}\cH_k$ are irreducible $\mathfrak{sl}_2$ lowest weight representations with weight vectors $R^{2j}\cH_k$ and lowest weight $k+M/2$. The weight vectors $R^{2j}\cH_k$ are $SO(m)\times Sp(2n)$-representations. However, in full superspace ($m\not=0\not= n$), these representations are not irreducible. This also implies that $\cP$ does not correspond to a multiplicity free irreducible direct sum decomposition for $ \mathfrak{sl}_2\times (SO(m)\times Sp(2n))$.

 In \cite{DBE1} the $SO(m)\times Sp(2n)$-module $\cH_k$ was decomposed into irreducible pieces. First, the following polynomials need to be introduced.
\begin{lemma}
If $0\le q \le n$ and $0\le k\le n-q$, there exists a homogeneous polynomial $f_{k,p,q}=f_{k,p,q}(r^2,\theta^2)$ (unique up to a multiplicative constant) of total degree $k$ such that $f_{k,p,q} \cH_p^b \otimes \cH_q^f \neq 0$ and $\Delta (f_{k,p,q} \cH_p^b \otimes \cH_q^f) = 0$.
This polynomial is given explicitly by
\[
 f_{k,p,q}=\sum_{s=0}^ka_sr^{2k-2s}\theta^{2s} \quad \mbox{with}\quad a_s=\binom{k}{s}\frac{(n-q-s)!}{\Gamma (\frac{m}{2}+p+k-s)}\frac{\Gamma(\frac{m}{2}+p+k)}{(n-q-k)!}.
\]
\label{polythm}
\end{lemma}

In particular, we find $f_{0,p,q}=1$. Using these polynomials we can obtain a full decomposition of the space of spherical harmonics of degree $k$.

\begin{theorem}[Decomposition of $\cH_k$]
Under the action of $SO(m) \times Sp(2n)$ the space $\cH_k$ decomposes as
\label{decompintoirreps}
\[
\cH_{k} = \bigoplus_{j=0}^{\min(n, k)} \bigoplus_{l=0}^{\min(n-j,\lfloor \frac{k-j}{2} \rfloor)} f_{l,k-2l-j,j} \cH^b_{k-2l-j} \otimes \cH^f_{j},
\]with $f_{l,k-2l-j,j}$ the polynomials determined in Lemma \ref{polythm}.
\end{theorem}

The corresponding projection operators are given in \cite{DBE1},
\begin{equation}
\label{projpiecesSphHarm}
\mQ_{r,s}^k= \prod_{i=0, \;  i \neq k-2r-s}^{k} \dfrac{\Delta_{LB,b} + i(m-2+i)}{(i-k+2r+s)(k+i-2r-s+m-2)}\times  \prod_{j=0, \;  j \neq s}^{\min{(n,k)}} \dfrac{\Delta_{LB,f} + j(-2n-2+j)}{(j-s)(j+s-2n-2)},
\end{equation}
with $\Delta_{LB,b}$ and $\Delta_{LB,f}$ as defined in equation \eqref{LB} or \eqref{LBosp} for the cases $n=0$ and $m=0$. They satisfy
\[
\mQ_{r,s}^k \left( f_{l,k-2l-j,j} \cH^b_{k-2l-j} \otimes \cH^f_{j} \right)= \delta_{rl} \delta_{sj} f_{l,k-2l-j,j} \cH^b_{k-2l-j} \otimes \cH^f_{j}.
\]

On a flat superspace, Berezin integration is defined by
\begin{eqnarray}
\label{superint}
\int_{\mR^{m | 2n}} = \int_{\mR^m} dV(\ux)\int_B=\pi^{-n}\int_{\mR^m} dV(\ux) \partial_{{x \grave{}}_{2n}} \ldots \partial_{{x \grave{}}_{1}} ,
\end{eqnarray}
with $dV(\ux)$ the usual Lebesgue measure on $\mR^{m}$.

In \cite{DBS9} the super Fourier transform on $\cS(\mR^m)\otimes \Lambda_{2n}\subset\cO_{\mR^{m|2n}}(\mR^m)$ was introduced as
\begin{eqnarray}
\label{Four}
\cF^{\pm}_{m|2n}(f(\bold{x}))(\bold{y})&=&(2\pi)^{-M/2}\int_{\mR^{m|2n},\bold{x}}\exp(\pm i\langle \bold{x},\bold{y}\rangle)f(\bold{x}),
\end{eqnarray}
with $\langle\bold{x},\bold{y}\rangle$ as defined in equation \eqref{inprod}. This Fourier transform satisfies
\begin{eqnarray}
\label{relFour}
\cF^{\pm}_{m|2n}\circ\nabla^2_{\bold{x}}=-R^2_{\bold{y}}\circ\cF^{\pm}_{m|2n}&\mbox{and}&\cF^{\pm}_{m|2n} \circ R^2_{\bold{x}}=-\nabla^2_{\bold{y}}\circ\cF^{\pm}_{m|2n}.
\end{eqnarray}

In case $m\not=0$, the supersphere is algebraically defined by the relation $R^2=1$. The integration over the supersphere was introduced in \cite{DBE1} for polynomials and generalized to a broader class of functions in \cite{MR2539324}. The uniqueness of this integration was also proven in \cite{DBE1, MR2539324}. 
\begin{theorem}
\label{SSintOxSp}
When $m\not=0$, the unique (up to a multiplicative constant) linear functional $T: \cP \rightarrow \mR$ satisfying the following properties for all $f(\bold{x}) \in \cP$:
\begin{itemize}
\item $T(R^2 f(\bold{x})) = T(f(\bold{x}))$
\item $T(f(S \cdot \bold{x})) = T(f(\bold{x}))$, \quad $\forall S \in SO(m)\times Sp(2n)$
\item $k \neq l \quad \Longrightarrow \quad T(\cH_k \cH_l) = 0 $
\end{itemize}
is given by the Pizzetti integral
\begin{eqnarray}
\label{Pizzetti}
\int_{\mS^{m-1|2n}} f(\bold{x})  =  \sum_{k=0}^{\infty}  \frac{2 \pi^{M/2}}{2^{2k} k!\Gamma(k+M/2)} (\nabla^{2k} f )(0)\quad \mbox{for}\quad f(\bold{x})\in\cP.
\end{eqnarray}
\end{theorem}
In the purely bosonic case, the third condition is not necessary for uniqueness, contrary to the full superspace case.

\section{The orthosymplectic supergroup and invariant functions}
\label{introOSp}

The dual pair $\left(SO(m)\times Sp(2n),\mathfrak{sl}_2\right)$ on $\mR^{m|2n}$ introduced above is not complete. This has several implications:
\begin{itemize}
\item \textbf{P1}: The weight vectors $R^{2j}\cH_k$ of $\mathfrak{sl}_2$ in the Fischer decomposition \eqref{superFischer} are not irreducible $SO(m)\times Sp(2n)$ representations, see Theorem \ref{decompintoirreps}. As a consequence the polynomials do not have a multiplicity free decomposition into irreducible pieces under the joint action $\mathfrak{sl}_2\times \left(SO(m)\times Sp(2n)\right)$.
\item \textbf{P2}: The only polynomials invariant under the action of the dual partner of $\mathfrak{sl}_2$ should be generated by the function amongst the generator of $\mathfrak{sl}_2$, $R^2$. However, all the elements of the commutative algebra generated by $r^2$ and $\theta^2$ are $SO(m)\times Sp(2n)$-invariant.
\item \textbf{P3}: The supersphere integration is not uniquely determined by the $SO(m)\times Sp(2n)$-invariance and the modulo $R^2-1$ property, see Theorem \ref{SSintOxSp}.
\end{itemize}
These problems will be solved by introducing the orthosymplectic supergroup $OSp(m|2n)$. This implies that the pair $(\mathfrak{osp}(m|2n),\mathfrak{sl}_2)$ or $(SOSp(m|2n),\mathfrak{sl}_2)$ is a true Howe dual pair for harmonic analysis on $\mR^{m|2n}$. The solutions to \textbf{P1}, \textbf{P2} and \textbf{P3} are given in Corollary \ref{OSprep}, Corollary \ref{Ospfun} and Theorem \ref{Pizuniekosp} respectively. The mathematical motivation to consider the orthosymplectic supergroup is given in the subsequent Theorem \ref{isomR}.

We will prove that invariance of functions and linear maps and the irreducibility of corepresentations can be expressed in terms of the super Harish-Chandra pair of a Lie supergroup. This is intuitively an immediate consequence of the equivalence of categories between Lie supergroups and Harish-Chandra pairs. The action of a Lie supergroup $\cG$ on a supermanifold $\cM$ induces an action of the Lie superalgebra $\mathfrak{g}$.

\begin{definition}
\label{defgamma}
Consider a Lie supergroup $\cG$ with action $\Psi$ on a supermanifold $\cM$. The action of an element $X\in \mathfrak{g}$ on an element $f\in\cO(\cM)$ is given by $(\mathfrak{g}, \cO(\cM))\to \cO(\cM)$, $(X,f)\to\gamma(X)f$, with $\gamma$ defined by
\begin{eqnarray*}
\gamma: \quad\mathfrak{g}\subset Der\cO(\cG)\to Der\cO(\cM)&;&X\to\gamma(X)=\left(\delta_{e_{\cG}}^{\sharp}\circ X\times id_{\cM}^\sharp\right)\circ \psi^{\sharp}.
\end{eqnarray*}
\end{definition}
The map $\gamma$ implies a Lie superalgebra morphism and can be extended to the universal enveloping algebra, these facts follow from straightforward calculations using equations \eqref{multid}, \eqref{multactie} and \eqref{invuniv}.

\begin{theorem}
\label{extUn}
The morphism of vector spaces $\mathfrak{g}\to Der(\cM)$: $X\to\gamma(X)$ is a Lie superalgebra morphism. Moreover, the action $\gamma( X)$ for $X\in \cU(\mathfrak{g})$ on $\cO(\cM)$ defined as
\begin{eqnarray}
\label{actieU}
\gamma( X)&=&\left(\delta_{e_{\cG}}^{\sharp}\circ X\times id_{\cM}^\sharp\right)\circ \psi^{\sharp}
\end{eqnarray}
satisfies $\gamma(X)\gamma(Y)=\gamma(XY)$ for $X,Y\in\cU(\mathfrak{g})$.
\end{theorem}

\begin{remark}
In \cite{MR2555976} the action of a Lie superalgebra on a supermanifold is defined in a similar setting. There it is also shown that this action and the action $\Psi_0=\Psi\circ(\delta_{\cG}\times id_{\cM})$ of $\cG_0$ on $\cM$ lead to an action of the Harish-Chandra pair on the supermanifold.
\end{remark}

The invariance under $\cG$ in Definition \ref{definvf} can be now be expressed in terms of $(\cG_0,\mathfrak{g})$. 
\begin{theorem}
\label{invfunctions}
Consider a superfunction $f\in\cO(\cM)$ and the action $\Psi:\cG\times\cM\to\cM$, of a Lie supergroup $\cG$, with Lie superalgebra $\mathfrak{g}$. The superfunction $f$ is $\cG$-invariant if and only if 
\begin{itemize}
\item $f$ is $\cG_0$-invariant, and
\item $\gamma(X)f=0$ for all $X\in \mathfrak{g}$.
\end{itemize}
\end{theorem}
\begin{proof}
One direction follows immediately from the definition of $\gamma(X)f$ and $\psi^\sharp_0=(\delta_{\cG}^\sharp\times id_{\cM}^\sharp)\circ\psi^\sharp$. 

Now we assume the invariance under $(\cG_0,\mathfrak{g})$ which can be rewritten as
\begin{eqnarray}
\label{thmnew1}
\left(\delta_{g}^\sharp\times id^\sharp_{\cM}\right)\circ\psi^\sharp (f)&=&f\quad \forall \,g\in\cG_0\quad\mbox{and}\\
\label{thmnew2}
\left(\xi_j\times id^\sharp_{\cM}\right)\circ\psi^\sharp (f)&=&0
\end{eqnarray}
where $\xi_j=\delta_0^\sharp \cD_j$ with $\cD_j$ introduced in Lemma \ref{defLiealg2}.

Similarly to Definition \ref{defLiealg1} and Lemma \ref{defLiealg2} we can define the Lie superalgebra $\mathfrak{g}'$ of left-invariant derivatives $Z\in$ $Der\cO(\cG)$, satisfying $\mu^\sharp\circ Z=(id_\cG^\sharp\times Z)\circ\mu^\sharp$. A basis is defined by $\{ (id_\cG^\sharp\times \xi_j)\circ\mu^\sharp\}$. This superalgebra $\mathfrak{g}'$ is isomorphic to the superalgebra of right-invariant derivatives $\mathfrak{g}$. Similarly to Definition \ref{defgamma} and Theorem \ref{extUn} there is a mapping $\gamma':\cU(\mathfrak{g}')\to\cU($Der$\cO(\cM)$), which now is an anti-algebra morphism. Because of this anti-algebra morphism property, equation \eqref{thmnew2} implies
that 
\begin{eqnarray*}
\left(\delta_{e_\cG}^\sharp\circ Z\times id_{\cM}^\sharp\right)\circ\psi^\sharp(f)&=&0
\end{eqnarray*}
holds for all $Z=Z_1\cdots Z_k$ with $k\ge 1$ and $Z_j\in\mathfrak{g}'$. Combining this with equations \eqref{multactie} and \eqref{thmnew1} yields
\begin{eqnarray*}
\left(\delta_{g}^\sharp\circ Z\times id_{\cM}^\sharp\right)\circ\psi^\sharp(f)&=&\left(\delta_{g}^\sharp\circ \left(id_\cG^\sharp\times\delta_{e_\cG}^\sharp\circ Z\right)\circ\mu^\sharp\times id_{\cM}^\sharp\right)\circ\psi^\sharp(f)\\
&=& \left(\delta_{e_\cG}^\sharp\circ Z\times id_{\cM}^\sharp\right)\circ\psi^\sharp(f)=0.
\end{eqnarray*}

Now consider a function $\alpha\in\cO(\cG)$ satisfying $\delta_g^\sharp(Z\alpha)=0$ for all $g\in\cG_0$ and $Z\in\cU(\mathfrak{g}')\backslash \mR$. The property $R_g^\sharp\circ R_{g^{-1}}^\sharp=id_\cG^\sharp$ (Proposition $2.2$ in \cite{MR0998351}) and $\delta_{e_\cG}^\sharp\circ R_g^\sharp=\delta_g^\sharp$ imply that this condition is equivalent with $\delta_{e_\cG}^\sharp(ZR_g^\sharp\alpha)=0$ for all $g\in\cG_0$ and $Z\in\cU(\mathfrak{g}')\backslash \mR$.

Equation \eqref{sheafsplit} implies that elements of $\cO(\cG)$ can be denoted as $f=\sum_Af_A(\uy)\uyb_A$ with $\{{\underline{y}\grave{}}_A\}$ a basis of monomials for $\Lambda_{2n}$ and $y_j$, $j=1,\cdots, p$ local coordinates on $\cG_0$. This corresponds to the supervector 
\begin{eqnarray*}
\bold{y}&=&(y_1,\cdots,y_p,{y\grave{}}_1,\cdots,{y\grave{}}_q)=(Y_1,\cdots, Y_{p+q}).
\end{eqnarray*}
A basis $\{Z^j\}$ for $\mathfrak{g}'$ can then be expressed in a neighborhood of the origin as
\begin{eqnarray*}
Z^j&=&\sum_{k=1}^{p+q}\beta^j_k(\bold{y})\frac{\partial}{\partial Y_k}
\end{eqnarray*}
for $j=1,\cdots,p+q$ with $\delta^\sharp_{e_\cG}\left(\beta^j_k\right)=\beta^j_k(0)=\delta^j_k$ and $|\beta^j_k|=|j|-|k|$. We know that $\alpha(\bold{y})\in\cO(\cG)$ introduced above satisfies $\delta^\sharp_{e_\cG}\circ X \alpha_g=0$, for all $X=Z^{j_1}Z^{j_2}\cdots Z^{j_k}$ with $j_i>p$ and $k>0$ and with $\alpha_g=R_g^\sharp(\alpha)$ for all $g\in\cG_0$. This implies $\delta^\sharp_{e_\cG}\partial_{{y\grave{}}_j}\alpha_g=0$ and then by induction
\begin{eqnarray*}
\delta^\sharp_{e_\cG}\left(\partial_{{y\grave{}}_{j_1}}\partial_{{y\grave{}}_{j_2}}\cdots\partial_{{y\grave{}}_{j_k}} \alpha_g\right)=0=\delta^\sharp_{g}\left(\partial_{{y\grave{}}_{j_1}}\partial_{{y\grave{}}_{j_2}}\cdots\partial_{{y\grave{}}_{j_k}} \alpha\right).
\end{eqnarray*}
This shows that $\alpha$ contains no fermionic variables or $\delta^\sharp_{\cG}(\alpha)=\alpha$. 

Applying this to our case yields
\begin{eqnarray*}
\psi^\sharp(f)&=&\left(\delta_{\cG}^\sharp\times id_{\cM}^\sharp\right)\circ\psi^\sharp(f)=1_{\cG}\times f,
\end{eqnarray*}
which proves the theorem.
\end{proof}

The orthogonal group is the matrix group which corresponds to the case $n=0$ in equation \eqref{OxSpdef}. The orthogonal group can also be characterized as the group of isometries of $\mR^m$ which stabilize the origin. Therefore we consider the supergroup of isometries of $\mR^{m|2n}$ which stabilize the origin. In \cite{MR2434470} the isometry supergroup of a Riemannian supermanifold was defined.

We consider the flat case $\mR^{m|2n}$ with metric $\langle \cdot,\cdot\rangle:Der \cO(\mR^{m|2n})\times Der \cO(\mR^{m|2n})\to\cO({\mR^{m|2n}})$ defined by
\begin{eqnarray*}
\langle f \nabla_j |h\nabla^k\rangle =(-1)^{|h|[j]}\delta_j^k fh
\end{eqnarray*}
for $f,h\in\cO({\mR^{m|2n}})$ and $h$ homogeneous in $\Lambda_{2n}$. This also implies $\langle \nabla^l,\nabla^k\rangle=g^{kl}$. The isometry supergroup of $\mR^{m|2n}$ is then defined as a Harish-Chandra pair $(\cG_0,\mathfrak{g})$. The Lie group $\cG_0$ is the group of supermanifold diffeomorphisms $\Phi:\mR^{m|2n}\to\mR^{m|2n}$ respecting the metric,
\[
\phi^\sharp(\langle Y^1_\phi|Y^2_\phi\rangle)=\langle Y^1,Y^2\rangle \mbox{ with }Y^j_\phi=(\phi^{\sharp})^{-1}\circ Y^j\circ\phi^\sharp  \mbox{ for all } Y^1,Y^2\in \mbox{Der}\cO({\mR^{m|2n}}).
\]
The Lie superalgebra $\mathfrak{g}$ is defined as the algebra of graded Killing vector fields. So $\mathfrak{g}$ is generated by all homogeneous vector fields $F$ such that
\begin{eqnarray}
\label{defKilling}
F\langle Y,Z\rangle&=&\langle [F,Y],Z\rangle+(-1)^{|F||Y|}\langle Y, [F,Z]\rangle
\end{eqnarray}
holds for all homogeneous vector fields $Y,Z$.

Since we will restrict to the supergroup of isometries which stabilize the origin, extra conditions on the Harish-Chandra pair need to be imposed. This implies that we should look for the supergroup $(\cH_0,\mathfrak{h})\subset (\cG_0,\mathfrak{g})$ such that the elements $\Phi$ of $\cH_0$ also satisfy $\delta_0^\sharp\circ\phi^\sharp=\delta_0^\sharp$ and the vector fields $F$ in $\mathfrak{h}$ also satisfy $\delta_0^\sharp\circ F=0$. With these definitions, the following theorem holds.

\begin{theorem}
\label{isomR}
The supergroup of isometries of the Riemannian superspace $\mR^{m|2n}$ which stabilize the origin is the Lie supergroup $OSp(m|2n)$, which corresponds to the Harish-Chandra pair 
\begin{eqnarray*}
\left(O(m)\times Sp(2n),\mathfrak{osp}(m|2n)\right).
\end{eqnarray*}
The action of $\mathfrak{osp}(m|2n)$ on $\mR^{m|2n}$ is given in equation \eqref{ospgen} and the action of $O(m)\times Sp(2n)$ is given in equation \eqref{actieOxSp}.
\end{theorem}

\begin{proof}
We denote the supergroup of isometries of $\mR^{m|2n}$ as defined in \cite{MR2434470} by $\cG=(\cG_0,\mathfrak{g})$. We calculate the supergroup $\cH=(\cH_0,\mathfrak{h})$ of isometries which stabilize the origin. It is straightforward to show that $O(m)\times Sp(2n)\subset\cH_0$ with the action of $O(m)\times Sp(2n)$ as defined in equation \eqref{actieOxSp}. Proposition $6$ in \cite{MR2434470} implies that an isometry of $\mR^{m|2n}$ (an element of $\cG_0$) is determined by its value and its derivative 
\begin{eqnarray*}
d\Phi: \mbox{Der}\cO(\mR^{m|2n})\to \mbox{Der}\cO(\mR^{m|2n}) &:&d\Phi(Y)={\phi^{-1}}^\sharp\circ Y\circ \phi^\sharp
\end{eqnarray*}
in one point. We take one $\Phi\in\cH_0$, this implies $\phi(0)=0$ since $\phi$ stabilizes the origin. The derivative in the origin is of the form $d_0\Phi(\nabla_k)=\sum_{j}a_{kj}\nabla_j$ with $a_{kj}$ constants. Since $\Phi$  preserves the metric, we immediately find that $(a_{jk})$ is an element of the matrix group $O(m)\times Sp(2n)$. The isometry $\Phi$ therefore has the same value and derivative in the origin as some isometry $\Phi'\in O(m)\times Sp(2n)$. Proposition $6$ in \cite{MR2434470} thus implies $\Phi=\Phi'$ and therefore $\cH_0=O(m)\times Sp(2n)$.

We consider condition \eqref{defKilling} for the Lie superalgebra in case $Y=\nabla^j$ and $Z=\nabla^k$ with $j,k=1,\cdots,m+2n$,
\begin{eqnarray*}
F\langle \nabla^j|\nabla^k\rangle&=&\langle [F,\nabla^j],\nabla^k\rangle+(-1)^{|F|[j]}\langle\nabla^j, [F,\nabla^k]\rangle.
\end{eqnarray*}
For a homogeneous vector field $F=\sum_{l=1}^{m+2n}F^l\nabla_l$ of degree $|F|$, the relation $[F,\nabla^j]=-(-1)^{|F|[j]}\sum_l\nabla^j(F^l)\nabla_l$ shows this condition is equivalent with
\begin{eqnarray*}
\nabla^j(F^k)+(-1)^{([j]+[k])|F|}(-1)^{[j][k]}\nabla^k(F^j)&=&0.
\end{eqnarray*}
This leads to
\begin{eqnarray*}
(\nabla^i\nabla^jF^k)&=&-(-1)^{([j]+[k])|F|}(-1)^{[j][k]+[i][k]}(\nabla^k\nabla^iF^j).
\end{eqnarray*}
Applying this consecutively for $(i,j,k)$, $(k,i,j)$ and $(j,k,i)$ yields $(\nabla^i\nabla^jF^k)=-(\nabla^i\nabla^jF^k)$. This implies that the functions $F^k$ are polynomials of maximal degree $1$. Since constants would lead to vector fields which do not stabilize the origin, the functions $F^k$ are elements of the vector space $\mC\{X_j|j=1,\cdots,m+2n\}$. An appropriate Killing vector field is therefore of the form $F=\sum_{kl}X_kF^{kl}\nabla_l$ where the coefficients satisfy
\begin{eqnarray*}
(-1)^{[j]}F^{jk}+(-1)^{([j]+[k])|F|}(-1)^{[j][k]}(-1)^{[k]}F^{kj}&=&0.
\end{eqnarray*}
By considering the cases $F$ even and odd separately we find that the derivatives corresponds to $\mathfrak{o}(m)\oplus \mathfrak{sp}(2n)$ or $\mathfrak{osp}(m|2n)_1$ in equation \eqref{ospgen}. It is straightforward to calculate that these $F$ satisfy the relation $F\langle Y,Z\rangle=\langle [F,Y],Z\rangle+(-1)^{|F||Y|}\langle Y, [F,Z]\rangle$ for general vector fields $Y$ and $Z$.
\end{proof}

\begin{remark}
Theorem \ref{isomR} implies that the action of $\mathfrak{osp}(m|2n)$ on $\mR^{m|2n}$ as defined in Definition \ref{defgamma} satisfies $\gamma(X_{ij})=L_{ij}$ with $L_{ij}$ given in equation \eqref{ospgen} for $X_{ij}$ the invariant derivatives on $OSp(m|2n)$.
\end{remark}

The Lie supergroup $SOSp(m|2n)$ corresponds to the Harish-Chandra pair $(SO(m)\times Sp(2n),\mathfrak{osp}(m|2n))$.
\begin{corollary}
\label{Ospfun}
The functions of the form $f(\bold{x})=h(R^2)\in\cC^\infty(\mR^m)\otimes\Lambda_{2n}$ with $h\in\cC^\infty(\mR^+)$, as in Definition \ref{fTaylor},
\begin{eqnarray*}
f(\bold{x})&=&h(R^2)=\sum_{j=0}^nh^{(j)}(r^2)\frac{\theta^{2j}}{j!},
\end{eqnarray*}
are the only functions on $\mR^{m|2n}$ which are $SOSp(m|2n)$-invariant.
\end{corollary}
\begin{proof}
In \cite{CDBS3} it was proven that these functions are the only superfunctions which satisfy $L_{ij} f=0$ for all $1\le i,j\le m+2n$. The $SO(m)\times Sp(2n)$-invariance is trivial.
\end{proof}
This corollary implies that the only invariants in $\cP$ are exactly $R^{2j}$ for $j\in\mN$. This solves problem \textbf{P2}. Applying Theorem \ref{invfunctions} shows that these functions are also invariant with respect to $OSp(m|2n)$.

The Laplace operator is also orthosymplectically invariant:
\begin{eqnarray*}
\psi^\sharp\circ\nabla^2&=&(id_{OSp(m|2n)}^\sharp\times \nabla^2)\circ\psi^\sharp,
\end{eqnarray*}
for $\Psi:OSp(m|2n)\times \mR^{m|2n}\to\mR^{m|2n}$, the action corresponding to theorem \ref{isomR}. This can again be easily checked by using the Harish-Chandra pair and the fact that $L_{ij}\nabla^2=\nabla^2 L_{ij}$ holds. Therefore we obtain that the generators of $\mathfrak{sl}_2$ are $OSp(m|2n)$-invariant which leads to the Howe dual pair $(\osp,\mathfrak{sl}_2)$ as a candidate for the Howe dual pair on $\mR^{m|2n}$.

\section{Invariant integration}

The invariance of linear maps (Definition \ref{definvmap}) can again be expressed in terms of the Harish-Chandra pair.
\begin{theorem}
\label{invmap}
Consider a supermanifold $\cM$ and a Lie supergroup $\cG$ with Lie superalgebra $\mathfrak{g}$ and action $\Psi$ on $\cM$. A linear map $T: \cO(\cM)\to \cV$ for some vector space $\cV$ is invariant with respect to the action $\Psi$ if and only if 
\begin{itemize}
\item $T$ is $\cG_0$-invariant, and
\item $T\circ\gamma(X)=0$ for all $X\in\mathfrak{g}$ with $\gamma$ given in Definition \ref{defgamma}.
\end{itemize}
\end{theorem}
\begin{proof}
First we assume $T$ is invariant. It is clear that $T$ is then also $\cG_0$-invariant. The identity map on $\cV$ is denoted $id_{\cV}$. For $X\in\mathfrak{g}$, Definition \ref{defgamma} and Definition \ref{definvmap} yield
\begin{eqnarray*}
T\circ \gamma(X)&=&\left(\delta_{e_{\cG}}^{\sharp}\circ X\times id_{\cV}\right)\circ\left(id_{\cG}^\sharp\times T\right)\circ \psi^{\sharp}\\
&=&\left(\delta_{e_{\cG}}^{\sharp}\circ X(1_{\cG})\times T\right)=0.
\end{eqnarray*} 

Now, assume $T$ is $\cG_0$-invariant, $\left(\delta_{\cG}^\sharp\times T\right)\circ \psi^{\sharp}=1_\cG\times T$ and $T\circ\gamma(X)=0$ holds for all $X\in\mathfrak{g}$. This property and Theorem \ref{extUn} yield
\begin{eqnarray*}
\left(\delta^\sharp_{e_\cG}\circ Y\times id_{\cV}\right)\circ \left(id_{\cG}^\sharp\times T\right)\circ\psi^\sharp=0
\end{eqnarray*}
for all $Y\in \cU(\mathfrak{g}')$ with $\mathfrak{g}'$ the left-invariant derivatives. Similarly to the proof of Theorem \ref{invfunctions}, this implies 
\begin{eqnarray*}
\left(id_{\cG}^\sharp\times T\right)\circ\psi^\sharp&=&(\delta^\sharp_{\cG}\times id_\cV)\circ \left(id_{\cG}^\sharp\times T\right)\circ\psi^\sharp,
\end{eqnarray*}
which proves the theorem.
\end{proof}

Before we deal with the integration over the supersphere in the next section we consider the case where the Lie supergroup and the supermanifold coincide, which corresponds to the problem in \cite{MR1845224, MR2172158}. In this case, the action $\Psi$ coincides with the multiplication $\mu$. The definition of $\gamma(X)$ in equation \eqref{actieU} for $X\in \cU(\mathfrak{g})$ and equations \eqref{invuniv} and \eqref{multid} then yield $\gamma(X)=X$. Theorem \ref{invmap} therefore implies that an integration $\int_{\cG}:\cO(\cG)\to\mR$ is left-invariant $\left((id_\cG^\sharp\times \int_{\cG})\circ \mu^\sharp=1_{\cG}\int_{\cG} \right)$ if and only if 
\begin{itemize}
\item $(\delta_{g}^\sharp\times \int_{\cG})\circ \mu^\sharp =\int_{\cG}\quad$ for all $g\in\cG_0$,
\item $\int_{\cG} \circ X=0\qquad\quad$ for all $X\in \mathfrak{g}$.
\end{itemize}

In \cite{MR1845224} it was proven that there is at most one (up to a multiplicative constant) invariant integration on a Lie supergroup. The construction of such an integral (for Lie supergroups such that the underlying Lie group is connected) in the setting of \cite{MR1845224, MR2172158} can be translated to our more analytical setting and for general Lie supergroups. These results are closely related to the supersphere integral of the next section since the supersphere $\mS^{m-1|2n}$ corresponds to the homogeneous supermanifold $OSp(m|2n)/OSp(m-1|2n)$. Invariant integration over $OSp(m|2n)$ can therefore be explicitly constructed from the supersphere integration of the next section,
\begin{eqnarray}
\int_{OSp(m|2n)}&=&\int_{\mS^{m-1|2n}}\circ\left(id^\sharp_{OSp(m|2n)}\times\int_{OSp(m-1|2n)}\circ I^\sharp\right)\circ\mu^\sharp.
\end{eqnarray}
In the equation above $I:OSp(m-1|2n)\to OSp(m|2n)$ is the embedding of $OSp(m-1|2n)$ in $OSp(m|2n)$. For a more mathematically rigorous approach to such Fubini type integrals, see \cite{MR2667819}.

Denote by $\int_{\cG_0}$ the unique left-invariant integration over the Lie group $\cG_0$. 

\begin{corollary}
Consider a Lie supergroup $\cG$ with Lie superalgebra $\mathfrak{g}$. Assume there exists an element $Y\in \cU(\mathfrak{g})$ such that 
\begin{itemize}
\item for all $X\in \mathfrak{g}$, $YX$ is an element of the right ideal $\mathfrak{g_0}\cU(\mathfrak{g})$,
\item for all $g\in\cG_0$, $Ad(g)(Y)-Y$ is an element of the right ideal $\mathfrak{g_0}\cU(\mathfrak{g})$.
\end{itemize}
Then the integration $\int_{\cG}=\int_{\cG_0}\circ\delta^\sharp_{\cG}\circ Y$ is left invariant. 
\end{corollary}
\begin{proof}
It is clear that $\int_{\cG_0}\circ\delta_{\cG}^\sharp$ is a $\cG_0$-invariant functional: $\left(\delta_\cG^\sharp\times \int_{\cG_0}\circ\delta_{\cG}^\sharp\right)\circ\mu^\sharp=1_{\cG}\int_{\cG_0}\delta_{\cG}^\sharp$. This implies that
\[\left(\delta_{g}^\sharp\times\int_{\cG_0}\circ\delta_{\cG}^\sharp\right)\circ\mu^\sharp=\int_{\cG_0}\circ\delta_{\cG}^\sharp \qquad\mbox{holds for all} \quad g\in\cG_0.\]
It also shows that $\int_{\cG_0}\circ\delta_{\cG}^\sharp\circ Z=0$ holds for all $Z\in\mathfrak{g_0}$. Using these results we can prove that the two conditions in Theorem \ref{invmap} are satisfied.

The second condition is satisfied, since
\begin{eqnarray*}
\int_{\cG}\circ X&=&\int_{\cG_0}\circ\delta^\sharp_{\cG}\circ YX=0
\end{eqnarray*}
for all $X\in \mathfrak{g}$, since $YX\in\mathfrak{g_0}\cU(\mathfrak{g})$. 

In order to prove the first condition we calculate for a general $g\in\cG_0$ and with $Ad$ the adjoint action as defined in formula \eqref{Adj},
\begin{eqnarray*}
\left(\delta_{g^{-1}}^\sharp\times\int_{\cG}\right)\circ\mu^\sharp&=&\int_{\cG_0}\circ\delta^\sharp_{\cG}\circ\left(\delta_{g^{-1}}^\sharp\times Y\right)\circ\mu^\sharp=\left(\delta_{g}^\sharp\times\int_{\cG_0}\circ\delta_{\cG}^\sharp\right)\circ\mu^\sharp\circ\left(\delta_{g^{-1}}^\sharp\times Y\right)\circ\mu^\sharp\\
&=&\int_{\cG_0}\circ\delta_{\cG}^\sharp\circ Ad(g)\left(Y\right)=\int_{\cG_0}\circ\delta_{\cG}^\sharp\circ Y.
\end{eqnarray*}
This proves the corollary.
\end{proof}

\section{The supersphere}

\subsection{The supersphere manifold}

The functions on the supersphere $\mS^{m-1|2n}$ are usually introduced as the elements of the superalgebra generated by $m+2n$ variables 
\begin{eqnarray*}
\{Y_1,\cdots,Y_m,Y_{m+1},\cdots,Y_{m+2n}\} &\mbox{with}& Y_i \mbox{ even (odd) if } i\le m \,\,(i>m),
\end{eqnarray*}
subject to the relation $\sum_{ij}Y_ig^{ij}Y_j=1$, with $g$ given in equation \eqref{defg}. In this paper we define the supersphere manifold in a more mathematically rigorous way.
\begin{definition}
\label{mathdefSS}
The supersphere manifold $\mS^{m-1|2n}$ corresponds to the supermanifold $(\mS^{m-1},\cC^\infty_{\mS^{m-1}}\otimes\Lambda_{2n})$ equipped with the super Riemannian metric $\langle \cdot,\cdot\rangle_{\mS^{m-1|2n}}:$ defined by
\begin{eqnarray*}
\langle L_{ij},L_{kl}\rangle_{\mS^{m-1|2n}}&=&(1-\theta^2_y)\langle L_{ij},L_{kl}\rangle_{\mS^{m-1}}\qquad \mbox{for} \quad 1\le i,j,k,l\le m\\
\langle L_{ij},\partial_{{y\grave{}}_l}\rangle_{\mS^{m-1|2n}}&=&0 \quad\mbox{and}\\
\langle \partial_{{y\grave{}}_k},\partial_{{y\grave{}}_l}\rangle_{\mS^{m-1|2n}}&=&J^{lk}+\frac{{y\grave{}}^k{y\grave{}}^l}{1-\theta_y^2}\qquad \mbox{for} \quad 1\le i,j\le m \quad\mbox{and}\quad 1\le k,l\le 2n.
\end{eqnarray*}
Here the ${y\grave{}}_k$ represent the generators of $\Lambda_{2n}$ and $L_{ij}$ are the standard vector fields on $\mS^{m-1}$ which generate $\mathfrak{o}(m)$, see also equation \eqref{ospgen}.
\end{definition}
From this definition it is clear that the structure sheaf of $\mS^{m-1|2n}$ is trivial, in fact we consider a product supermanifold $\mS^{m-1|2n}\cong \mS^{m-1}\otimes \mR^{0|2n}$. The fact that the metric is not the product metric makes the Riemannian superspace $\mS^{m-1|2n}$ a non-trivial object.

Then we can define the functions $Y_k\in\cO(\mS^{m-1|2n})$, $k=1,\cdots,m+2n$ as 
\[
Y_j=\begin{cases}
\sqrt{1-\theta^2_y}\,\xi_j& \mbox{for } 1\le j\le m\\
{y\grave{}}_{j-m} &\mbox{for } m+1\le j\le m+2n.
\end{cases}
\]
The functions $\xi_i$ on $\mS^{m-1}$ are determined by $\xi_i(\uv)=v_i$ for $\uv$ a vector on $\mS^{m-1}$. The functions $Y_j$ clearly satisfy 
\begin{eqnarray*}
\sum_{i,j=1}^{m+2n}Y_ig^{ij}Y_j&=&1
\end{eqnarray*}
and define an embedding of the supersphere in $\mR^{m|2n}$. This embedding will be considered in more detail in subsection \ref{subsemb}. The metric defined above is also the induced metric corresponding to this embedding.

\begin{lemma}
The supergroup of isometries of the supersphere $\mS^{m-1|2n}$ is $OSp(m|2n)$. The action of $\mathfrak{osp}(m|2n)$ is given by
\begin{eqnarray}
\label{ospgen2}
J_{ij}=\begin{cases}
L_{ij} & i,j\le m\\
L_{ij}^{(y)}={y\grave{}}_{i-m}\partial_{{y\grave{}}^{j-m}} +{y\grave{}}_{j-m}\partial_{{y\grave{}}^{i-m}} & i,j >m\\
\xi_i\sqrt{1-\theta^2_y}\partial_{{y\grave{}}^{j-m}}-\frac{{y\grave{}}_{j-m}}{\sqrt{1-\theta^2_y}}\sum_{l=1}^m\xi_lL_{li} & i\le m < j.
\end{cases}
\end{eqnarray}
and the action of $O(m)\times Sp(2n)$ is the classical action of $O(m)$ on $\mS^{m-1}$ and of $Sp(2n)$ on $\Lambda_{2n}$.
\end{lemma}
\begin{remark}
\label{remosp}
The vector fields $J_{ij}$ can be symbolically denoted by $Y_i\partial_{Y^j}-(-1)^{[i][j]}Y_j\partial_{Y^i}$ when evaluated on functions on $\mS^{m-1|2n}$ which are contained in the algebra generated by the functions $Y_k$.
\end{remark}
\begin{proof}
Remark \ref{remosp} implies that the vector fields $J_{ij}$ generate $\mathfrak{osp}(m|2n)$. The fact that the vector fields $J_{ij}$ correspond to Killing vector fields can be calculated directly or derived immediately from the subsequent Theorem \ref{achterosp}. The fact that these are the only Killing vector fields can also be calculated directly or derived from the embedding in the next section.
\end{proof}

\begin{theorem}
\label{Pizuniekosp1}
When $m\not=0$, the unique (up to a multiplicative constant) $SOSp(m|2n)$-invariant linear functional on $\cO(\mS^{m-1|2n})$ is given by
\begin{eqnarray*}
\int_{\mS^{m-1|2n}}\cdot&=&\int_{\mS^{m-1}}\int_B \left(1-\theta^2_y\right)^{\frac{m}{2}-1}.
\end{eqnarray*}
\end{theorem}
\begin{proof}
All linear functionals on $\Lambda_{2n}$ can be realized as 
\begin{eqnarray*}
T:\Lambda_{2n}\to\mR&:& \beta(\uyb)\to T(\beta(\uyb)) =\int_{B,y}\alpha(\uyb)\beta(\uyb),
\end{eqnarray*}
with $\alpha(\uyb)\in\Lambda_{2n}$. If $T$ is $Sp(2n)$-invariant, $\alpha(\uyb)$ is also $Sp(2n)$-invariant. Therefore, the only $SO(m)\times Sp(2n)$-invariant linear functionals on $\cC^\infty(\mS^{m-1})\otimes\Lambda_{2n}$ are of the form
\begin{eqnarray*}
\int_{\mS^{m-1}}\int_B\alpha(\theta^2_y)\cdot.
\end{eqnarray*}
According to Theorem \ref{invmap}, such a functional is $SOSp(m|2n)$-invariant if and only if
\begin{eqnarray*}
\int_{\mS^{m-1}}\int_B\alpha(\theta^2_y)J_{ij}\cdot&=&0
\end{eqnarray*}
for all $1\le i\le j\le m+2n$. This is clearly already satisfied for the even generators. Imposing this condition for $J_{1,m+1}=2\xi_1\sqrt{1-\theta^2_y}\partial_{{y\grave{}}_{2}}-\frac{{y\grave{}}_{1}}{\sqrt{1-\theta^2_y}}\sum_{l=1}^m\xi_lL_{l1}$ is equivalent with
\begin{eqnarray*}
-2\int_{\mS^{m-1}}\xi_1\int_B\partial_{{y\grave{}}_{2}}\left(\alpha(\theta^2_y)\sqrt{1-\theta^2_y}\right)\cdot&=&(m-1)\int_{\mS^{m-1}}\xi_1\int_B\alpha(\theta^2_y)\frac{{y\grave{}}_{1}}{\sqrt{1-\theta^2_y}}\cdot.
\end{eqnarray*}
This is satisfied if and only if
\begin{eqnarray*}
\partial_{{y\grave{}}_{2}}\left(\alpha(\theta^2_y)\sqrt{1-\theta^2_y}\right)&=&-\frac{m-1}{2}\alpha(\theta^2_y)\frac{{y\grave{}}_{1}}{\sqrt{1-\theta^2_y}}
\end{eqnarray*}
holds. This uniquely determines $\alpha(\theta^2_y)$ to be $(1-\theta^2_y)^{\frac{m}{2}-1}$. It is clear that $\int_{\mS^{m-1}}\int_B\alpha(\theta^2_y)J_{ij}\cdot=0$ will also hold for the other $J_{ij}$.
\end{proof}

The case $m=0$ will be considered in the subsequent Remark \ref{remarkm0}.

\begin{remark}
The unicity and existence of the integral is implied by the theory in \cite{MR2667819}. There it is proven that the existence is equivalent with the fact that Ber$\left(\osp/\mathfrak{osp}(m-1|2n)\right)^\ast$ is a trivial $OSp(m-1|2n)$-module.
\end{remark}

The formula for the integral has a logical interpretation in terms of the superdeterminant of the metric. For any set of local coordinates $(t_1,\cdots,t_{m-1})$ on $\mS^{m-1}$ and the global coordinates $({y\grave{}}_{1},\cdots,{y\grave{}}_{2n})$, the metric $h$ of $\mS^{m-1|2n}$ is given by
\begin{eqnarray*}
h&=&\left( \begin{array}{c|c} A&0\\ \hline \vspace{-3.5mm} \\0&B
\end{array}
 \right)
\end{eqnarray*}
where $A=(1-\theta^2_y)h_0$ with $h_0$ the metric on $\mS^{m-1}$ with respect to $(t_1,\cdots,t_{m-1})$. The matrix $B$ follows immediately from Definition \ref{mathdefSS}. A technical calculation shows that $\det B=(1-\theta^2_y)/2^{2n}$. This implies that 
\begin{eqnarray*}
\sqrt{s\det(h)}&=&\sqrt{\det(A)\det(B)^{-1}}=(1-\theta^2_y)^{\frac{m}{2}-1}2^n\sqrt{\det(h_0)},
\end{eqnarray*}
which explains the appearance of the factor $(1-\theta^2_y)^{\frac{m}{2}-1}$ in the expression in Theorem \ref{Pizuniekosp1}.

\subsection{Embedding of the supersphere}
\label{subsemb}

In \cite{Mehler}, \cite{DBS9} and \cite{DBE1}, restricting polynomials on $\mR^{m|2n}$ to the supersphere was realized by taking $\cP$ modulo the ideal generated by $R^2-1$. This is purely algebraical and therefore not applicable to general superfunctions. It is also only practical if $M\not\in-2\mN$, which is the case where a Fischer decomposition (\ref{superFischer}) exists. In this section we will construct a different and more rigorous interpretation of the embedding of the supersphere. 

To generalize the properties of the unit sphere integration in section \ref{classHarm}, we need to define a condition `$f(\bold{x})$ is zero on the supersphere' for $f\in \cC^\infty(\mR^m)\otimes\Lambda_{2n}$. The condition $R^2=L^2$ with $L\in\mR$ is equivalent with $r^2=L^2-\theta^2$. Define the superalgebra morphism $\phi^\sharp:\cC^\infty(\mR^{m})\otimes\Lambda_{2n}\to\cC^\infty(\mR^{m}_0)\otimes\Lambda_{2n}$ (with $\mR^m_0=\mR^m\backslash \{0\}$) as
\begin{eqnarray*}
\phi^\sharp(f(\bold{x}))&=&\sum_{j=0}^n\frac{(-1)^j\theta^{2j}}{j!}\left(\frac{\partial}{\partial r^2}\right)^jf(\bold{x}).\\
\end{eqnarray*}
This map $\phi^\sharp$ can be identified with a morphism of sheaves (which we also denote by $\phi^\sharp$) corresponding to the supermanifold diffeomorphism $\Phi=(id_{\mR^m},\phi^\sharp):\mR^{m|2n}_0=(\mR^m_0,\cC^\infty_{\mR^m_0}\otimes\Lambda_{2n})\to\mR^{m|2n}$. In the sense of Definition \ref{fTaylor}, $\phi^\sharp$ corresponds to substituting $r^2-\theta^2$ for $r^2$. Substituting $L^2-\theta^2$ for $r^2$ is therefore realized by $[\phi^\sharp(f(\bold{x}))]_{r=L}$. The condition $f=0$ for $R^2=L^2$ is thus expressed as $[\phi^\sharp(f(\bold{x}))]_{r=L}=0$. In particular for functions of the form in Corollary \ref{Ospfun}, this definition yields the expected result:
\begin{eqnarray*}
[\phi^\sharp(h(R^2))]_{r=L}&=&h(L^2).
\end{eqnarray*}

According to Definition \ref{fTaylor}, $1/R\in\cO(\mR^{m|2n}_0)$ is defined as
\begin{eqnarray*}
\frac{1}{R}&=&\sum_{j=0}^n(-1)^j\frac{\theta^{2j}}{r^{1+2j}}\frac{\Gamma(j+\frac{1}{2})}{\Gamma(\frac{1}{2})j!}.
\end{eqnarray*}

\begin{lemma}
\label{Alg}
Denote by $\tilde{\Lambda}_{2n}$ the subalgebra of $\Lambda_{2n}\otimes \mR\{\frac{1}{r}\}\subset \cO(\mR^{m|2n}_0)$ generated by $\{\frac{{x\grave{}}_j}{R}|j=1,\cdots, 2n\}$. The algebra $\tilde{\Lambda}_{2n}$ is isomorphic to $\Lambda_{2n}$.
\end{lemma}
\begin{proof}
The morphism $\xi:\Lambda_{2n}\to \tilde{\Lambda}_{2n}$ given by ${x\grave{}}_A\to {x\grave{}}_A/R^{|A|}$ is surjective. To prove the injectivity we construct the inverse morphism. Define $\chi_i: \tilde{\Lambda}_{2n}    \to \Lambda_{2n} $ for $i=0,\cdots,n$ by
\begin{eqnarray*}
\chi_0(a)&=&\lim_{r\to\infty}a\\
\chi_i(a)&=&\chi_{i-1}(a)+\lim_{r\to\infty}r^i\left[a-\xi(\chi_{i-1}(a))\right]\qquad\mbox{for } i=1,\cdots,n.
\end{eqnarray*}
Then we calculate for $g=\sum_Ag_A{x\grave{}}_A\in \Lambda_{2n}$, $\chi_0(\xi(g))=g_{0,\cdots,0}=\sum_{A,|A|\le0}g_A{x\grave{}}_A$. Assuming $\chi_i(\xi(g))=\sum_{A,|A|\le i}g_A{x\grave{}}_A$, yields
\begin{eqnarray*}
\chi_{i+1}(\xi(g))=\lim_{r\to\infty}r^{i+1}\left[\sum_{A,|A|> i}g_A\frac{{x\grave{}}_A}{r^{|A|}}\frac{1}{\left(\sqrt{1+\frac{\theta^2}{r^2}}\right)^{|A|}}\right]+\sum_{A,|A|\le i}g_A{x\grave{}}_A=\sum_{A,|A|\le i+1}g_A{x\grave{}}_A,
\end{eqnarray*}
so by induction, $\chi_n$ is the inverse of $\xi$.
\end{proof}

The isomorphism $\xi$ from the proof of this theorem can be trivially to an isomorphism on $\cO(\mR^{m|2n}_0)$ by letting it act trivially on bosonic functions.

\begin{definition}
\label{evalu}
Consider an open subset $U^\ast$ of $\mR^m$ and $U=\pi^{-1}_{\mS^{m-1}}(U^\ast)$. The evaluation on the supersphere $\pi_{\mS^{m-1|2n}}^\sharp$ is a superalgebra morphism $\cO^\infty_{\mR^{m|2n}}(U^\ast)\to \cC^\infty_{\mS^{m-1}}(U)\otimes \tilde{\Lambda}_{2n}$ given by $\pi_{\mS^{m-1|2n}}^\sharp=\xi\circ(\pi^\sharp_{\mS^{m-1}}\times id_{\Lambda_{2n}})\circ\phi^\sharp$.
\end{definition}
The evaluation on the supersphere for $f(\bold{x})=\sum_Af_A(\ux){x\grave{}}_A\in\cC^\infty(\mR^m)\otimes \Lambda_{2n}$ is therefore given by
\begin{eqnarray*}
\pi_{\mS^{m-1|2n}}^\sharp f&=&\sum_{A}\sum_{j=0}^{n}(-1)^j\left[\left(\frac{\partial}{\partial r^2}\right)^jf_A(\ux)\right]_{r=1}\frac{\theta^{2j}}{j!R^{2j}}\frac{{x\grave{}}_A}{R^{|A|}}.
\end{eqnarray*}

This definition behaves well with respect to the definition of `$f$ is zero on the supersphere' since $[\phi^\sharp]_{r=1}=0$ holds if and only if $\pi^\sharp_{\mS^{m-1|2n}}[f]=0$ holds. The following lemma proves Definition \ref{evalu} is a natural choice for the supersphere evaluation.
\begin{lemma}
\label{lemSSev}
The evaluation on the supersphere on $\cO(\mR^{m|2n})$ can be expressed using Definition \ref{fTaylor} as
\begin{eqnarray*}
\pi_{\mS^{m-1|2n}}^\sharp f&=&f(\frac{\bold{x}}{R})\,=\,\sum_{\underline{\alpha}}\sum_Af_A^{(\ua)}(\frac{\ux}{r})\frac{\ux^{\ua}}{(\ua !)}(\frac{1}{R}-\frac{1}{r})^{|\ua |}\frac{{x\grave{}}_A}{R^{|A|}}.
\end{eqnarray*}
\end{lemma}

\begin{proof}
By Definition \ref{fTaylor}, for $f\in\cC^\infty(\mR^m)$,
\[f\left(\frac{\ux}{R}\right)=f\left(\frac{\ux}{r}+\ux(\frac{1}{R}-\frac{1}{r})\right)=\sum_{\underline{\alpha}}f^{(\ua)}(\uom)\frac{\ux^{\ua}}{(\ua !)}(\frac{1}{R}-\frac{1}{r})^{|\ua |}
\]
holds since $\left(\frac{1}{R}-\frac{1}{r}\right)$ is nilpotent. Now we calculate
\begin{eqnarray*}
\frac{1}{k!}\left(\frac{\partial^k f}{\partial r^k}\right)(\uom)&=&\frac{1}{k!}\left[\sum_{\ua,|\ua|=k} f^{(\ua)}(\ux)\underline{\omega}^{\ua}\frac{|\ua| !}{(\ua!)}\right]_{r=1}=\sum_{\ua,|\ua|=k} f^{(\ua)}(\uom)\frac{\ux^{\ua}}{r^k}\frac{1}{(\ua!)},
\end{eqnarray*}
which implies
\begin{eqnarray*}
f\left(\frac{\ux}{R}\right)&=&\sum_{k=0}^{n}\frac{1}{k!}\left[\frac{\partial^k }{\partial r^k}f(\ux)\right]_{r=1}\left(\frac{r}{R}-1\right)^{k}.
\end{eqnarray*}

Now, since the relation
\[
\sum_{j=0}^n\left[\left(\frac{\partial}{\partial r^2}\right)^jf(\ux)\right]_{r=1}(-1)^j\frac{\theta^{2j}}{j!R^{2j}}=\sum_{j=0}^n\left[\left(\frac{\partial}{\partial r}\right)^jf(\ux)\right]_{r=1}\frac{1}{j!}\left(\sqrt{1-\frac{\theta^2}{R^2}}-1\right)^j
\]
holds, the lemma is proven.
\end{proof}

According to Lemma \ref{Alg} we can identify ${x\grave{}}_j/R$ with ${y\grave{}}_j$, where ${x\grave{}}_j$  are the generators of the Grassmann algebra for $\mR^{m|2n}$ and ${y\grave{}}_j$ those of the Grassmann algebra for $\mS^{m-1|2n}$. This implies that $\Pi_{\mS^{m-1|2n}}=(\pi_{\mS^{m-1}},\pi^\sharp_{\mS^{m-1|2n}})$ is a supermanifold morphism $\mS^{m-1|2n}\to\mR^{m|2n}$ and the supersphere is embedded in the Euclidean superspace. Because of the identification of ${x\grave{}}_j/R$ with ${y\grave{}}_j$ we find 
\begin{eqnarray}
\label{SSevsimple}
\pi^\sharp_{\mS^{m-1|2n}}=(\pi^\sharp_{\mS^{m-1}}\times \rho_{YX})\circ\phi^\sharp: & &\cO(\mR^{m|2n})\to\cO(\mS^{m-1|2n}),
\end{eqnarray}
with $\rho_{YX}$ the identity map between the Grassmann algebra generated by the ${x\grave{}}_j$ and ${y\grave{}}_j$.

Now we can study supersphere integration for functions on $\mR^{m|2n}$. First we need the following technical lemma.
\begin{lemma}
\label{OO}
The operator $\phi^\sharp:\cO(\mR^{m|2n})\to\cO(\mR_0^{m|2n})$ behaves with respect to the coordinates and derivatives as
\begin{eqnarray*}
\phi^\sharp x_j=x_j\sqrt{1-\frac{\theta^2}{r^2}}\phi^\sharp,& &\phi^\sharp{x\grave{}}_j={x\grave{}}_j\phi^\sharp,\\
\phi^\sharp\partial_{x_j}=\frac{1}{\sqrt{1-\frac{\theta^2}{r^2}}}\partial_{x_j}\phi^\sharp-\frac{x_j\theta^2}{r^3\sqrt{1-\frac{\theta^2}{r^2}}}\partial_{r}\phi^\sharp&\mbox{and} &\phi^\sharp\partial_{{x\grave{}}_{j}}=\partial_{{x\grave{}}_{j}}\phi^\sharp-\frac{{x\grave{}}^{j}}{r}\partial_{r}\phi^\sharp.
\end{eqnarray*}
\end{lemma}
\begin{proof}
Since $\phi^\sharp$ is a superalgebra morphism, $\phi^\sharp( X_j f(\bold{x}))=\phi^\sharp (X_j)\phi^\sharp(f(\bold{x}))$ holds. The calculation of the expressions $\phi^\sharp (x_j)$ and $\phi^\sharp ({x\grave{}}_j)$ are straightforward. To calculate the third property we use the fact that $\partial_{r^2}$ and $L_{ij}$ for $1\le i,j\le m$ commute with $\phi^\sharp$ and the first property in the lemma, leading to
\begin{eqnarray*}
\partial_{x_j}\phi^\sharp&=&\left[\frac{x_j}{r}\partial_r+\sum_{l=1}^m\frac{x_l}{r^2}L_{lj}\right]\phi^\sharp\\
&=&\frac{x_j}{r}\partial_r\phi^\sharp+\sum_{l=1}^m\sqrt{1-\frac{\theta^2}{r^2}}\,\phi^\sharp \,\frac{x_l}{r^2}L_{lj}\\
&=&\frac{x_j}{r}\partial_r\phi^\sharp+\sqrt{1-\frac{\theta^2}{r^2}}\phi^\sharp \,\left[\partial_{x_j}-\frac{x_j}{r}\partial_r\right]\\
&=&\frac{x_j}{r}\left(1-(1-\frac{\theta^2}{r^2})\right)\partial_r\phi^\sharp+\sqrt{1-\frac{\theta^2}{r^2}}\phi^\sharp \,\partial_{x_j}.\\
\end{eqnarray*}
The last property follows from a straightforward calculation.
\end{proof}

Now we prove that the metric on $\mS^{m-1|2n}$ in Definition \ref{mathdefSS} is the induced metric from the embedding of $\mS^{m-1|2n}$ in $\mR^{m|2n}$ given by $\Pi_{\mS^{m-1|2n}}$.

\begin{theorem}
\label{achterosp}
The realizations of $\mathfrak{osp}(m|2n)$ in equations \eqref{ospgen} and \eqref{ospgen2} satisfy
\begin{eqnarray*}
(i)& &J_{ij}\circ\pi^\sharp_{\mS^{m-1|2n}}=\pi^\sharp_{\mS^{m-1|2n}}\circ L_{ij}.
\end{eqnarray*}
The generators $J_{ij}$ of $\mathfrak{osp}(m|2n)$ in \eqref{ospgen2} which span the tangent space of $\mS^{m-1|2n}$ and the metric $\langle,\rangle_{\mS^{m-1|2n}}$ in Definition \ref{mathdefSS} satisfy
\begin{eqnarray*}
(ii)& &\langle J_{ij},J_{kl}\rangle_{\mS^{m-1|2n}}=\pi^\sharp_{\mS^{m-1|2n}}(\langle L_{ij},L_{kl}\rangle)\qquad \mbox{for} \quad 1\le i,j,k,l\le m+2n.
\end{eqnarray*}

\end{theorem}
\begin{proof}
The first statement for the even generators of $\mathfrak{osp}(m|2n)$ is trivial. Lemma \ref{OO} implies
\begin{eqnarray*}
& &\phi^\sharp \circ L_{i,m+j}\\
&=&\left[x_i\sqrt{1-\frac{\theta^2}{r^2}}\partial_{{x\grave{}}^j} +2x_i\sqrt{1-\frac{\theta^2}{r^2}}\frac{{x\grave{}}_j}{2r}\partial_{r}\ -{x\grave{}}_j\frac{1}{\sqrt{1-\frac{\theta^2}{r^2}}}\partial_{x_i} +{x\grave{}}_j\frac{x_i\theta^2}{r^3\sqrt{1-\frac{\theta^2}{r^2}}}\partial_{r}\right]\circ\phi^\sharp \\
&=&x_i\sqrt{1-\frac{\theta^2}{r^2}}\partial_{{x\grave{}}^j}\circ\phi^\sharp -{x\grave{}}_j\frac{1}{\sqrt{1-\frac{\theta^2}{r^2}}}\partial_{x_i}\circ\phi^\sharp +\frac{{x\grave{}}_jx_i}{r\sqrt{1-\frac{\theta^2}{r^2}}}\partial_{r}\circ\phi^\sharp \\
&=&x_i\sqrt{1-\frac{\theta^2}{r^2}}\partial_{{x\grave{}}^j}\circ\phi^\sharp -\frac{{x\grave{}}_j}{r^2\sqrt{1-\frac{\theta^2}{r^2}}}\sum_{k=2}^mx_kL_{ki}\circ\phi^\sharp.
\end{eqnarray*}
Equation \eqref{SSevsimple} then yields
\begin{eqnarray*}
\pi_{\mS^{m-1|2n}}^\sharp\circ L_{i,m+j}&=&\xi_i\sqrt{1-\theta^2_y}\partial_{{y\grave{}}^j}\circ \pi^\sharp_{\mS^{m-1|2n}} -\frac{{y\grave{}}_j}{\sqrt{1-\theta^2_y}}\sum_{k=2}^m\xi_kL_{ki}\circ\pi_{\mS^{m-1|2n}}^\sharp,
\end{eqnarray*}
which proves the first part of the theorem. 

The second statement is straightforward for the generators of $\mathfrak{o}(m)$. Lemma \ref{OO} implies that $\partial_{{y\grave{}}_j}\circ \pi^\sharp_{\mS^{m-1|2n}}=\pi^\sharp_{\mS^{m-1|2n}}\circ \left(\partial_{{x\grave{}}_j}+\frac{{x\grave{}}^j}{r}\partial_r\right)$ holds. Combining this with
\begin{eqnarray*}
\langle L_{ij},\partial_{{y\grave{}}_k}\rangle_{\mS^{m-1|2n}}&=&0=\pi^\sharp_{\mS^{m-1|2n}}(\langle L_{ij},\left(\partial_{{x\grave{}}_k}+\frac{{x\grave{}}^k}{r}\partial_r\right)\rangle) \quad \mbox{and}\\
\langle \partial_{{y\grave{}}_k},\partial_{{y\grave{}}_l}\rangle_{\mS^{m-1|2n}}&=&\pi^\sharp_{\mS^{m-1|2n}}(\langle \left(\partial_{{x\grave{}}_k}+\frac{{x\grave{}}^k}{r}\partial_r\right),\left(\partial_{{x\grave{}}_l}+\frac{{x\grave{}}^l}{r}\partial_r\right)\rangle)
\end{eqnarray*}
for $1\le i,j\le m$ and $1\le k,l\le 2n$ then leads to the proposed equations.
\end{proof}

Now we can prove the uniqueness of the supersphere integration for functions on $\mR^{m|2n}$, thus solving problem \textbf{P3}. This gives a natural generalization of the characterization of the integral over the sphere in subsection \ref{classHarm}.
\begin{theorem}
\label{Pizuniekosp}
When $m\not=0$, the only (up to a multiplicative function) linear map
\begin{eqnarray*}
T: \cO(\mR^{m|2n}) \rightarrow \cC^\infty(\mR^+); & &f(\bold{x})\to T[f](L),
\end{eqnarray*}
satisfying the properties
\begin{itemize}
\item $T[R^2 f](L)=L^2 \, T[f](L)$, 
\item $T$ is $SOSp(m|2n)$-invariant,
\end{itemize}
is given by
\begin{eqnarray*}
\int_{\mS^{m-1|2n}(L)}\cdot&=&\int_{\mS^{m-1}(L)}\int_B \left(1-\frac{\theta^2}{r^2}\right)^{\frac{m}{2}-1}\phi^\sharp\cdot\\
&=&\frac{1}{L^{m-2}}\int_{\mS^{m-1}(L)}\int_B \phi^\sharp\, r^{m-2}\,\cdot
\end{eqnarray*}
\end{theorem}

\begin{proof}
If $[\phi^\sharp(f)]_{r=L}=0$ then $\phi^\sharp(f)=(r^2-L^2)h$ for some $h\in \cO(\mR^{m|2n}_0)$. Therefore $f=\left(\phi^\sharp\right)^{-1}(r^2-L^2)\left(\phi^\sharp\right)^{-1}(h)=(R^2-L^2)\left(\phi^\sharp\right)^{-1}(h)$ and $T[f](L)=0$.

This implies that $T[f](L)$ only depends on $\left[\phi^\sharp f\right]_{r=L}$. The theorem is then proven by observing that $\left[\phi^\sharp f\right]_{r=L}$ can be identified with $L^{-\mE_f}\pi^\sharp_{\mS^{m-1|2n}}[f(L\bold{x})]$ and applying theorems \ref{Pizuniekosp1} and \ref{achterosp}.
\end{proof}

Theorem \ref{invmap} implies that this integration is also $OSp(m|2n)$-invariant. 

\begin{remark}
\label{remarkm0}
In case $m=0$ the formulation of Theorem \ref{Pizuniekosp} does not make sense. It is however immediately clear that there is a unique (up to a multiplicative function) linear map $T:\Lambda_{2n}\to\cC^\infty(\mR^+)$, such that 
\begin{itemize}
\item $T[\theta^2 f](L)=L^2\, T[f](L)$,
\item $T$ is $Sp(2n)$-invariant. 
\end{itemize}
This is a direct consequence of the decomposition of $\Lambda_{2n}$ in Lemma \ref{superFischerLemma} which is the decomposition into irreducible pieces under the action of $Sp(2n)$, see \cite{MR1153249}.
\end{remark}

The expression \eqref{LBosp} and the invariance of the integral immediately imply
\begin{eqnarray}
\label{SSLB}
\int_{\mS^{m-1|2n}(L)}\circ\,\Delta_{LB}&=&0.
\end{eqnarray}

When $L=1$, the notation $\int_{\mS^{m-1|2n}}=\int_{\mS^{m-1|2n}(1)}$ is also used for supersphere integration on functions on $\mR^{m|2n}$. Evaluated on polynomials, $\int_{\mS^{m-1|2n}}$ corresponds to the Pizzetti integral in equation \eqref{Pizzetti}, see Theorem 8 in \cite{MR2539324}. By using $\phi^\sharp(f(L\bold{x}))=(\phi^\sharp f)(L\bold{x})$ and $\int_B \beta(L\uxb)=L^{2n}\int_B \beta(\uxb)$ we obtain
\begin{eqnarray*}
L^{M-1}\int_{\mS^{m-1|2n}}f(L\bold{x})&=&\int_{\mS^{m-1|2n}(L)}f(\bold{x})
\end{eqnarray*}
for the choice of normalization in Theorem \ref{Pizuniekosp}. The supersphere integral on $\mR^{m|2n}$ can also be written as
\begin{eqnarray*}
\int_{\mS^{m-1|2n}(L)}f(\bold{x})&=&2L^{2-m}\int_{\mR^{m|2n}}\delta(R^2-L^2)f(\bold{x}).
\end{eqnarray*}
by introducing the distribution
\begin{eqnarray*}
\delta(R^2-L^2)&=&\sum_{j=0}^n\delta^{(j)}(r^2-L^2)\frac{\theta^{2j}}{j!}\quad\in\;\cC^\infty(\mR^m)'\otimes\Lambda_{2n},
\end{eqnarray*}
see Lemma 10 in \cite{MR2539324}.

\subsection{The basic spherical mean}
In this section we introduce the basic spherical mean in superspace.

\begin{definition}
\label{defSS}
For $f\in \cO(\mR^{m|2n})$, the super basic spherical mean is given by
\begin{eqnarray*}
Mf(\bold{x},L)&=&\int_{\mS^{m-1|2n},\bold{y}}f(\bold{x}+L\bold{y}).
\end{eqnarray*}
\end{definition}

The basic spherical mean is a function in $m+1|2n$ variables, while the original function only has $m|2n$ variables. So we could expect $Mf$ to satisfy supplementary relations. This is the subject of the following theorem.

\begin{theorem}(Darboux)
For $f\in \cO(\mR^{m|2n})$, $Mf$ satisfies the super differential equation of Darboux:
\begin{eqnarray*}
\left[\nabla^2_{\bold{x}}-\frac{\partial^2}{\partial L^2}-\frac{M-1}{L}\frac{\partial}{\partial L}\right]Mf(\bold{x},L)&=&0.
\end{eqnarray*}
\label{Darboux}
\end{theorem}

\begin{proof}
The first condition in theorem \ref{Pizuniekosp} implies $\int_{\mS^{m-1|2n}}R^2\cdot=\int_{\mS^{m-1|2n}}\cdot$. We calculate $\nabla^2_{\bold{x}}Mf(\bold{x},L)$ using equations (\ref{LB}) and \eqref{SSLB},
\begin{eqnarray*}
\nabla^2_{\bold{x}}Mf(\bold{x},L)&=&\frac{1}{L^2}\int_{\mS^{m-1|2n},\bold{y}}\nabla^2_{\bold{y}}f(\bold{x}+L\bold{y})\\
&=&\frac{1}{L^2}\int_{\mS^{m-1|2n},\bold{y}}[\Delta_{LB,\bold{y}}+\mE_{\bold{y}}(M-2+\mE_{\bold{y}})]f(\bold{x}+L\bold{y})\\
&=&\frac{1}{L^2}\int_{\mS^{m-1|2n},\bold{y}}L\frac{\partial}{\partial L}(M-2+L\frac{\partial}{\partial L})f(\bold{x}+L\bold{y}),
\end{eqnarray*}
proving the theorem.
\end{proof}

Using the Funk-Hecke theorem in superspace, see \cite{CDBS3}, we can prove that the Pizzetti formula \eqref{Pizzetti} for integration over the supersphere also holds for functions in the Schwartz space $\cS(\mR^{m})\otimes\Lambda_{2n}\subset\cO(\mR^{m|2n})$.

\begin{theorem}
If $\varphi\in\cS(\mR^{m})\otimes\Lambda_{2n}$ and $M>1$, then formally
\begin{eqnarray*}
M\varphi(\bold{x},L)&=& \sum_{k=0}^{\infty} \frac{2 \pi^{M/2}L^{2k}}{2^{2k} k!\Gamma(k+M/2)} \nabla^{2k} \varphi(\bold{x}).
\end{eqnarray*}
\end{theorem}

\begin{proof}
Since the super Fourier transform \eqref{Four} is an isomorphism on $\cS(\mR^{m})\otimes\Lambda_{2n}$, see \cite{DBS9}, for every $\varphi\in\cS(\mR^{m})\otimes\Lambda_{2n}$ there is an $\phi\in\cS(\mR^{m})\otimes\Lambda_{2n}$ such that
\begin{eqnarray*}
\varphi(\bold{x})&=&\int_{\mR^{m|2n},\bold{z}}\exp(i\langle \bold{x},\bold{z}\rangle)\phi(\bold{z}).
\end{eqnarray*}
From this we obtain the basic spherical mean using theorem 8 in \cite{CDBS3},
\begin{eqnarray*}
M\varphi(\bold{x},L)&=&\int_{\mR^{m|2n},\bold{z}}\exp(i\langle \bold{x},\bold{z}\rangle)\phi(\bold{z})\int_{\mS^{m-1|2n},\bold{y}}\exp(iL\langle\bold{ y},\bold{z}\rangle)\\
&=&\int_{\mR^{m|2n},\bold{z}}\exp(i\langle \bold{x},\bold{z}\rangle)\phi(\bold{z})(2\pi)^{\frac{M}{2}}(\sqrt{R_{\bold{z}}^2}L)^{1-\frac{M}{2}}J_{\frac{M}{2}-1}(\sqrt{R_{\bold{z}}^2}L)\\
&=&2\pi^{M/2}\int_{\mR^{m|2n},\bold{z}}\exp(i\langle \bold{x},\bold{z}\rangle) \sum_{k=0}^{\infty} \frac{(-1)^kL^{2k}R^{2k}_{\bold{z}}}{2^{2k} k!\Gamma(k+M/2)}  \phi(\bold{z})
\end{eqnarray*}
where $J_{\alpha}$ represents the Bessel function of the first kind. This is formally equal to the proposed equation because of the properties \eqref{relFour} of the super Fourier transform.
\end{proof}

\section{Polynomials on the supersphere}

Similarly to Theorem \ref{invfunctions} and \ref{invmap}, we can prove that a corepresentation of $\cG$ is irreducible if and only if the corresponding representation \eqref{indHCrep} of $(\cG_0,\mathfrak{g})$ is irreducible. Since the theorem is not surprising and the proof uses very similar ideas and techniques as in the proof of Theorem \ref{invfunctions} we omit it.

\begin{theorem}
\label{equivirr}
Consider a Lie supergroup $\cG$ and a super vectorspace $\cV$. A corepresentation $\chi$ of $\cG$ (Definition \ref{defrep1}) on $\cV$ is irreducible if and only if the representation of the super Harish-Chandra pair of $\cG$ (Definition \ref{defrep2}) as constructed in equation \eqref{indHCrep} is irreducible.
\end{theorem}

The action $\Psi$ of a Lie supergroup on a supermanifold defines a corepresentation (which corresponds to $\psi^\sharp$) of the supergroup on the vector space of superfunctions $\cO(\cM)$. The definition of $\rho^\pi$ in equation \eqref{indHCrep} then coincides with the definition of $\gamma$ in Definition \ref{defgamma}. From the form of $\gamma$ for $\cG=OSp(m|2n)$ and $\cM=\mR^{m|2n}$ given in equation \eqref{ospgen} we find that this induces a representation of $\mathfrak{osp}(m|2n)$ on the space of polynomials $\cP\subset\cO(\mR^{m|2n})$. The same holds for the action of $GL(m|2n)$ on $\mR^{m|2n}$.

\subsection{Polynomials on super Euclidean space and on the supersphere}

The polynomials on $\mR^{m|2n}$ form a $\mathfrak{gl}(m|2n)$-module, the action of the basis $E_{ij}$ of $\mathfrak{gl}(m|2n)$ on $\mR^{m|2n}$ is given by
\begin{eqnarray*}
\gamma(E_{ij})&=&X_i\partial_{X_j}\qquad \mbox{for}\quad i,j=1\cdots,m+2n.
\end{eqnarray*}
The decomposition $\cP=\bigoplus_{k=0}^\infty\cP_k$ is the decomposition of the polynomials on $\mR^{m|2n}$ into irreducible pieces under the action of $\mathfrak{gl}(m|2n)$, where $\cP_k$ is the representation with highest weight $(k,0,\cdots,0)$. In what follows we prove the corresponding results for polynomials on the supersphere $\mS^{m-1|2n}$ and the Lie superalgebra $\mathfrak{osp}(m|2n)$. 

If $m-2n\not\in-2\mN$, the polynomials on the supersphere can be identified with the harmonic polynomials on $\mR^{m|2n}$,
\begin{eqnarray}
\label{decompH}
\pi_{\mS^{m-1|2n}}^\sharp\cP\cong\cH&=&\bigoplus_{k=0}^\infty \cH_k.
\end{eqnarray}
This immediately follows from the Fischer decomposition \eqref{superFischer} and the fact 
\begin{eqnarray*}
\pi^\sharp_{\mS^{m-1|2n}}(R^{2j}H_k)&=&H_k/R^k
\end{eqnarray*} 
for $H_k\in\cH_k$. The action of $\mathfrak{osp}(m|2n)$ on polynomials is given by the expressions in formula \eqref{ospgen}. Since the $L_{ij}$ commute with $\mE$ and $\nabla^2$, the blocks $\cH_k$ are $\mathfrak{osp}(m|2n)$-representations. In appendix A of \cite{MR2395482} the irreducibility of $\cH_k$ as an $\mathfrak{osp}(m|2n)$-representation was proven for $M=m-2n>1$. Some very specific examples were also proven in proposition $3.1$ in \cite{MR2441598}. In the next section we will generalize this result to all $M\not\in-2\mN$, which implies that the formula \eqref{decompH} is the decomposition of $\cH$ into irreducible blocks under the action of $\mathfrak{osp}(m|2n)$. This will also solve problem $\textbf{P1}$ and complete the interpretation of the Fischer decomposition \eqref{superFischer}. We will also consider the case $M\in-2\mN$ which leads to unexpected and interesting results.

The supersphere integration in theorem \ref{Pizuniekosp} or equation \eqref{Pizzetti} also generates a superhermitian bilinear $\langle\cdot|\cdot\rangle_{\mS^{m-1|2n}}$ form on each space $\cH_k$, (or on $\cH$ or $\pi_{\mS^{m-1|2n}}^\sharp\cP$), by defining
\begin{eqnarray*}
\langle f|g\rangle_{\mS^{m-1|2n}}&=&\int_{\mS^{m-1|2n}}f\overline{g}.
\end{eqnarray*}
This implies $\langle f|g\rangle_{\mS^{m-1|2n}}=(-1)^{|f||g|}\overline{\langle g|f\rangle}_{\mS^{m-1|2n}}$, for $f$ and $g$ homogeneous.

The representation of $\mathfrak{osp}(m|2n)$ on $\cH$ is an orthosymplectic representation (see definition in \cite{MR0922822}) with respect to this bilinear form. For $f$ and $g\in\cH_k$ with $f$ homogeneous,
\begin{eqnarray*}
\langle L_{ij}f|g\rangle_{\mS^{m-1|2n}}&=&\int_{\mS^{m-1|2n}}L_{ij}f\overline{g}-(-1)^{([i]+[j])|f|}\int_{\mS^{m-1|2n}}fL_{ij}\overline{g}\\
&=&-(-1)^{([i]+[j])|f|}\langle f|L_{ij}g\rangle_{\mS^{m-1|2n}}.
\end{eqnarray*}

\subsection{The $\cH_k$-representations}

First we need the following technical lemma.
\begin{lemma}
The functions introduced in Lemma \ref{polythm} satisfy the relation
\label{Lf}
\begin{eqnarray*}
L_{i,2j-1+m}f_{k,p,q}&=&2k(\frac{M}{2}+p+q+k-1)f_{k-1,p+1,q+1}x_i{x\grave{}}_{2j-1}
\end{eqnarray*}
for $i\le m$ and $j\le n$ and $L_{il}$ defined in equation \eqref{ospgen}.
\end{lemma}
\begin{proof}
This expression follows from a direct calculation. Since\\$L_{i,2j-1+m}=\left(2x_i\partial_{{x\grave{}}_{2j}}-{x\grave{}}_{2j-1}\partial_{x_i}\right)$, the relation
\begin{eqnarray*}
L_{i,2j-1+m}f_{k,p,q}
&=&2x_i{x\grave{}}_{2j-1}\sum_{s=1}^k\frac{k!}{(k-s)!(s-1)!}\frac{(n-q-s)!}{\Gamma (\frac{m}{2}+p+k-s)}\frac{\Gamma(\frac{m}{2}+p+k)}{(n-q-k)!}r^{2k-2s}\theta^{2s-2}\\
&-&{x\grave{}}_{2j-1}2x_i\sum_{s=0}^{k-1}\frac{k!}{(k-s-1)!s!}\frac{(n-q-s)!}{\Gamma (\frac{m}{2}+p+k-s)}\frac{\Gamma(\frac{m}{2}+p+k)}{(n-q-k)!}r^{2k-2s-2}\theta^{2s}\\
&=&2k(\frac{M}{2}+p+q+k-1)x_i{x\grave{}}_{2j-1}f_{k-1,p+1,q+1}
\end{eqnarray*}
holds, which proves the lemma.
\end{proof}

The decomposition of $\cH_k$ in Theorem \ref{decompintoirreps} can be represented by the following diagram, where we use the notation $(l,k-2l-j,j)\leftrightarrow f_{l,k-2l-j,j}\cH_{k-2l-j}^b\otimes\cH_j^f$.
\[
\xymatrix@=10pt@!C{
(0,k,0)&(0,k-1,1)&(0,k-2,2)&(0,k-3,3)&\cdots\\
&(1,k-2,0)&(1,k-3,1)\ar[u]_{\textbf{(1)}}\ar[d]_{\textbf{(4)}}\ar[r]_{\textbf{(3)}}\ar[l]_{\textbf{(2)}}&(1,k-4,2)&\cdots\\
&&(2,k-4,0)&(2,k-5,1)&\cdots\\
}
\]
The arrows represent the actions of the elements of $\cU(\mathfrak{osp}(m|2n))$ constructed in the following lemma, where we will use the notation $(l,p,q)\leftrightarrow f_{l,p,q}\cH_{p}^b\cH_q^f$ again. It is important for the sequel to note that the decomposition above is finite, which is clear from the fact that $\cH_k$ is finite dimensional.

\begin{lemma}
\label{arrows}
Consider the spherical harmonics on $\mR^{m|2n}$ with $m\not=0\not= n$. For $j=1,\cdots,4$, there are spherical harmonics $e_j\in (l,p,q)\not= \emptyset$ (so $l+q\le n$ and $p\le 1$ when $m=1$) and elements $u_j\in \cU(\mathfrak{osp}(m|2n))$ such that $u_je_j$ is non-zero and
\[
u_1 e_1\in(l-1,p+1,q+1)\qquad \mbox{if} \qquad p+q+l\not=1-\frac{M}{2},
\]
\[
u_2 e_2\in (l,p+1,q-1), \quad u_3 e_3\in (l,p-1,q+1) \quad \mbox{and}\quad  u_4 e_4\in (l+1,p-1,q-1)\]
if the resulting subspace of spherical harmonics is different from the empty set.
\end{lemma}
\begin{proof}
The Laplace-Beltrami operators $\Delta_{LB,b}$ and $\Delta_{LB,f}$ are Casimir operators of respectively $\mathfrak{o}(m)$ and $\mathfrak{sp}(2n)$, see formula \eqref{LBosp} for the limit cases $n=0$ and $m=0$. In particular they are elements of $\cU(\osp)$. This means the projection operators in equation \eqref{projpiecesSphHarm} are elements of $\cU(\mathfrak{osp}(m|2n))$. 

First, consider for $q<n$, a spherical harmonic $H_q^f\in\cH_q^f$ which does not contain ${x\grave{}}_1$ or ${x\grave{}}_2$. This exists since it can be taken as a spherical harmonic of degree $q$ in the Grassmann algebra $\Lambda_{2n-2}$ generated by the $\{{x\grave{}}_j,j>2\}$ (and $q\le n-1$). This also implies ${x\grave{}}_1H_q^f\in\cH_{q+1}^f$ and different from zero. Next, consider a spherical harmonic $H_p^b\in \cH_p^b$ (with $p\le 1$ when $m=1)$. Define $H_{p+1}^b\in \cH_{p+1}^b$ by
\begin{eqnarray*}
H_{p+1}^b&=&x_1H_p^b-\frac{r^2}{2p+m-2}\partial_{x_1}H_p^b.
\end{eqnarray*}
Lemma \ref{Lf} then implies
\begin{eqnarray*}
L_{1,1+m}f_{l,p,q}H_p^bH_q^f&=&2l(\frac{M}{2}+p+q+l-1)\,f_{l-1,p+1,q+1}\, H_{p+1}^b\, {x\grave{}}_{1}H_q^f\\
&+&\frac{l}{p+\frac{m}{2}-1}(\frac{M}{2}+p+q+l-1)\,r^2\,f_{l-1,p+1,q+1}\,\partial_{x_1}H_{p}^b\, {x\grave{}}_{1}H_q^f\\
&-&f_{l,p,q}\,\partial_{x_1}H_{p}^b\, {x\grave{}}_{1}H_q^f.
\end{eqnarray*}
The sum of the last two lines is of the form $g(r^2,\theta^2)\partial_{x_1}H_{p}^b{x\grave{}}_{1}H_q^f$ and is an element of $\cH_{2l+p+q}$ since the left-hand and right-hand side of the first line are. Lemma \ref{polythm} then yields $g(r^2,\theta^2)=\lambda_1 f_{l,p-1,q+1}$ for some $\lambda_{1}\in\mR$, with $\lambda_1\not=0$ since otherwise $f_{l,p,q}\equiv 0$ mod $r^2$. The identity $\frac{l}{p+\frac{m}{2}-1}(\frac{M}{2}+p+q+l-1)r^2f_{l-1,p+1,q+1}-f_{l,p,q}=\lambda_1 f_{l,p-1,q+1}$ also follows from a direct calculation, showing that $\lambda_1=-(1+\frac{l}{p+\frac{m}{2}-1})$. Summarizing, we obtain
\begin{eqnarray*}
L_{1,1+m}f_{l,p,q}H_p^bH_q^f&=&2l(\frac{M}{2}+p+q+l-1)f_{l-1,p+1,q+1}H_{p+1}^b{x\grave{}}_{1}H_q^f\\
&+&\lambda_1 f_{l,p-1,q+1}(\partial_{x_1}H_{p}^b)({x\grave{}}_{1}H_q^f)
\end{eqnarray*}
for some $\lambda_1\not=0$. Using this we can prove the arrows \textbf{(1)} and \textbf{(3)}. Both arrows start from spaces with $q<n$ since otherwise $\cH_{q+1}^f=\emptyset$.

\textbf{(1)}: Take $H_p^b$ and $H_q^f$ as defined above and assume $x_1H_p^b\not=\frac{r^2}{2p}\partial_{x_1}H_p^b$. This is always possible, the only non-trivial case is if $m=1$, but then $p$ has to be zero since otherwise $\cH_{p+1}^b=\emptyset$. The previous calculations then imply
\begin{eqnarray*}
\mQ^{2l+p+q}_{l-1,q+1}\left(L_{1,1+m}f_{l,p,q}H_p^bH_q^f\right)&=&2l(\frac{M}{2}+p+q+l-1)f_{l-1,p+1,q+1}H_{p+1}^b{x\grave{}}_{1}H_q^f\\
&\in& f_{l-1,p+1,q+1}\cH_{p+1}^b\otimes\cH_{q+1}^f
\end{eqnarray*}
which is different from zero since $l>0$ (otherwise $f_{l-1,p+1,q+1}\cH_{p+1}^b\otimes\cH_{q+1}^f=\emptyset$) and since we assumed $p+q+l\not=1-\frac{M}{2}$ for this arrow.

\textbf{(3)}: Take again $H_p^b$ and $H_q^f$ as defined above but now assume $\partial_{x_1}H_p^b\not=0$ (which is always possible since $p>0$ for this arrow). We then obtain
\begin{eqnarray*}
\mQ^{2l+p+q}_{l,q+1}\left(L_{1,1+m}f_{l,p,q}H_p^bH_q^f\right)&=&\lambda_1 f_{l,p-1,q+1}(\partial_{x_1}H_{p}^b)({x\grave{}}_{1}H_q^f)\\
&\in& f_{l,p-1,q+1}\cH_{p-1}^b\otimes\cH_{q+1}^f,
\end{eqnarray*}
which is different from zero.

The arrows \textbf{(2)} and \textbf{(4)} can be proven similarly. Consider for $1\le q\le n$ a spherical harmonic $\cH_{q-1}^f\in \cH_{q-1}^f$ which does not contain ${x\grave{}}_{1}$ or  ${x\grave{}}_{2}$, define $H_{q+1}^f\in\cH_{q+1}^f$ (not necessarily different from zero) by
\begin{eqnarray*}
H_{q+1}^f&=&{x\grave{}}_{1}{x\grave{}}_{2}H_{q-1}^f-\frac{1}{q-n-1}\theta^2H_{q-1}^f.
\end{eqnarray*}
Then we calculate
\begin{eqnarray*}
L_{1,1+m}f_{l,p,q}H_p^b{x\grave{}}_2H_{q-1}^f&=&2l(\frac{M}{2}+p+q+l-1)f_{l-1,p+1,q+1}x_1H_p^b{x\grave{}}_{1}{x\grave{}}_{2}H_{q-1}^f\\
&+&2f_{l,p,q}x_1H_p^bH_{q-1}^f-f_{l,p,q}\partial_{x_1}H_p^b{x\grave{}}_{1}{x\grave{}}_{2}H_{q-1}^f.
\end{eqnarray*}

This leads to
\begin{eqnarray*}
\mQ^{2l+p+q}_{l,q-1}\left(L_{1,1+m}f_{l,p,q}H_p^b{x\grave{}}_2H_{q-1}^f\right)=\lambda_2 f_{l,p+1,q-1}H_{p+1}^bH_{q-1}^f\in f_{l,p+1,q-1}\cH_{p+1}^b\cH_{q-1}^f.
\end{eqnarray*}
with $\lambda_2 f_{l,p+1,q-1}=2f_{l,p,q}-\frac{2l}{q-n-1}(\frac{M}{2}+p+q+l-1)\theta^2f_{l-1,p+1,q+1}$ with $\lambda_2\not=0$ since otherwise $f_{l,p,q}\equiv 0$ mod $\theta^2$.

The calculation above also implies
\begin{eqnarray*}
& &\mQ^{2l+p+q}_{l+1,q-1}\left(L_{1,1+m}f_{l,p,q}H_p^b{x\grave{}}_2H_{q-1}^f\right)\\
&=&\lambda_3 f_{l+1,p-1,q-1}(\partial_{x_1}H_{p}^b)H_{q-1}^f\in f_{l+1,p-1,q-1}\cH_{p-1}^b\cH_{q-1}^f,
\end{eqnarray*}
with
\begin{eqnarray*}
\lambda_3 f_{l+1,p-1,q-1}&=&\frac{l(\frac{M}{2}+p+q+l-1)}{(p+\frac{m}{2}-1)(q-n-1)}f_{l-1,p+1,q+1}\theta^2r^2+\frac{r^2}{p+\frac{m}{2}-1}f_{l,p,q}-\frac{\theta^2}{q-n-1}f_{l,p,q}.
\end{eqnarray*}

Again, $\lambda_3\not=0$, since otherwise $r^2f_{l,p,q}\equiv 0$ mod $\theta^2$. This implies arrows \textbf{(2)} and \textbf{(4)} by identifying $e_2=\mQ^{2l+p+q}_{l,q-1}\circ L_{1,1+m}$ and $e_4=\mQ^{2l+p+q}_{l+1,q-1}\circ L_{1,1+m}$.
\end{proof}

\begin{remark}
\label{irrEnv}
A representation $(\mathfrak{g},V)$ of a finite dimensional Lie superalgebra $\mathfrak{g}$ on a finite dimensional vector space is irreducible if and only if for each two vectors $(u,v)\in V$ there is an element $X$ in the universal enveloping algebra $\cU(\mathfrak{g})$ such that $Xu=v$.
\end{remark}

\begin{theorem}
\label{irrH}
When $M=m-2n\not\in-2\mN$, the space $\cH_k$ is an irreducible $\mathfrak{osp}(m|2n)$-module. When $M\in-2\mN$, $\cH_k$ is irreducible if and only if 
\begin{eqnarray*}
k>2-M &\mbox{or}& k< 2-\frac{M}{2}.
\end{eqnarray*}
\end{theorem}

\begin{proof}
For $M\in-2\mN$, the $\mathfrak{sl}_2$ relations for $\nabla^2$, $R^2$ and $\mE+\frac{M}{2}$, see \eqref{sl2rel} (or Lemma 3 in \cite{DBE1}), imply that for $2-\frac{M}{2}\le k\le 2-M$
\begin{eqnarray}
\label{submodule}
R^{2k+M-2}\cH_{2-M-k}\subset \cH_k.
\end{eqnarray}
This is a submodule (since $\dim\cH_{2-M-k}< \dim \cH_{k}$) which proves the reducibility for $2-\frac{M}{2}\le k\le 2-M $. 

For the other cases, $M\not\in-2\mN$ or $M\in-2\mN$ and $k$ not in the interval $[2-\frac{M}{2},2-M]$, we consider $f_{l,p,q}\cH_p^b\otimes\cH_q^f\subset\cH_k$, so $k=2l+p+q$. If $M\not\in-2\mN$, $p+q+l+\frac{M}{2}-1$ is never zero. If $M\in-2\mN$, $p+q+l=1-\frac{M}{2}$ implies  (since $l>0$ and $p+q\ge0$)
\begin{eqnarray*}
k=2l+p+q\ge 2-\frac{M}{2}&\mbox{and}&k=2l+p+q\le 2-M.
\end{eqnarray*}
This implies that the arrows in Lemma \ref{arrows} will always exist for $M\not\in-2\mN$ and for $M\in-2\mN$ with $k\not\in[2-\frac{M}{2},2-M]$. 

We prove the theorem using Remark \ref{irrEnv}. First we assume $m\not=1$ and take one fixed $H_k^b\in\cH_k^b\subset\cH_k$. It suffices to prove that for every element of the form\begin{eqnarray*}
u&=&f_{l,k-2l-j}H_{k-2l-j}^bH_j^f\in f_{l,k-2l-j}\cH_{k-2l-j}^b\otimes\cH_j^f\\
\end{eqnarray*}
there are elements $X,Y\in \cU(\mathfrak{osp}(m|2n))$ for which $XH_k^b=u$ and $Yu=H_k^b$ . Since $(\mathfrak{o}(m),\cH_k^b)$ is irreducible there is an element $X_1\in \cU(\mathfrak{o}(m))$ such that $X_1H_k^b$ is an element in $\cH_k^b$ for which there is an element $X_2$ of $\cU(\mathfrak{osp}(m|2n))$ determined in Lemma \ref{arrows} (arrow \textbf{(3)}) such that
\begin{eqnarray*}
X_2X_1H_k^b&=&H_{k-1}^bH_1^f\in\cH_{k-1}^b\otimes\cH_1^f.
\end{eqnarray*}
This procedure can be repeated until we find an $X'\in \cU(\mathfrak{osp}(m|2n))$ and an $H_{k-j}^bH_j^f\in\cH_{k-j}^b\otimes\cH_j^f$ such that
\begin{eqnarray*}
X'H_k^b&=&H_{k-j}^bH_j^f.
\end{eqnarray*}
By the same arguments and by using arrow \textbf{(4)} in Lemma \ref{arrows} we find an $X\in \cU(\mathfrak{osp}(m|2n))$ such that $XH_k^b=u$. The proof for $Y$ is similar, by using arrows \textbf{(1)} and \textbf{(2)}.

Now consider the case $m=1$, when $k\le n$. Taking into account that $H_p^b=\emptyset$ when $p>1$, the decomposition diagram above Lemma \ref{arrows} looks like
\[
\xymatrix@=7pt@!C{
&&(0,1,k-1)\ar[d]_{\textbf{(4)}}&(0,k,0)\ar[l]_{\textbf{(2)}}\\
&\cdots&(1,0,k-2)\\
&(\lfloor\frac{k}{2}\rfloor-1,1,\nu(k)+1)\ar[d]_{\textbf{(4)}}&(\lfloor\frac{k}{2}\rfloor-1,0,\nu(k)+2)\ar[l]_{\textbf{(2)}}&\\
\left[(\frac{k-1}{2},1,0)\right]&(\lfloor\frac{k}{2}\rfloor,0,\nu(k))\ar[l]_{\textbf{(2)}}\\
}
\]

where $\nu(k)$ is equal to $0$ if $k$ is even and equal to $1$ if $k$ is odd. The last part is between square brackets since it only exists if $k$ is odd. From this diagram it is immediately clear that $\cH_k$ will be irreducible. The case $n<k\le 2n+1$ can be treated completely similarly.
\end{proof}

\begin{remark}
The results Lemma \ref{arrows} and the technique in the proof of Theorem \ref{irrH} show that $\cH_k$ (for all values of $k$ and $m|2n$) always corresponds to a Verma module or the quotient of a Verma module. This implies that even in case $\cH_k$ is not irreducible, it will not be decomposable. This will be explained in more detail in Corollary \ref{notdecomp}.
\end{remark}

The results in Theorem \ref{irrH} and Theorem \ref{equivirr} yield the following corollary about the irreducibility of corepresentation (Definition \ref{defrep1}) and representation (Definition \ref{defrep2}) of $SOSp(m|2n)$ (or $OSp(m|2n)$) on the spaces of spherical harmonics.
\begin{corollary}
\label{OSprep}
The space $\cH_k$ of spherical harmonics of degree $k$ on $\mR^{m|2n}$ is an irreducible $SOSp(m|2n)$-(co)representation if $M\not\in-2\mN$.
\end{corollary}
This solves problem \textbf{P1}.

We introduce the notation $L_\lambda^{m|2n}$ for the unique irreducible $\osp$-module with highest weight $\lambda$, see e.g. \cite{MR051963}.  The irreducible $\mathfrak{sl}_2$-module with lowest weight $j$ is denoted by $L^{(j)}$.

\subsection{The case $m=1$}
For the simple root system of $\mathfrak{osp}(1|2n)$ we choose the standard simple root system of $\mathfrak{sp}(2n)$. The only non-zero spaces of spherical harmonics are $\cH_k$ for $k\le 2n+1$, since $\cH_p^b=\emptyset$ for $p>1$. If $k\le n$, the highest weight vector of $\cH_k$ for $\mathfrak{osp}(1|2n)$ is the highest weight vector of $\cH_k^f$ for $\mathfrak{sp}(2n)$. For $k> n$ (and $k\le 2n+1$), the highest weight vector is the highest weight vector of 
\begin{eqnarray*}
f_{k-n-1,1,2n-k+1}\,x_1\,\cH_{2n-k+1}^f
\end{eqnarray*}
for $\mathfrak{sp}(2n)$. This can be seen from the fact that $(k-n-1,2n-k+1)$ is the element of all pairs $(j,l)$ subject to $j+l\le n$ and $2j+l=k$ or $2j+l+1=k$ for which $l$ is maximal. This implies that the highest weight of $\cH_k$ is $(1,\cdots,1,0,\cdots,0)$ where the integer $1$  is repeated $k$ times for $k\le n$ and $2n-k+1$ times for $k>n$. Hence, we obtain
\begin{eqnarray*}
\cH_k&\cong&L^{1|2n}_{(1,\cdots,1,0,\cdots,0)}
\end{eqnarray*}
with $L^{1|2n}_{(1,\cdots,1,0,\cdots,0)}$ the irreducible $\mathfrak{osp}(1|2n)$ representation with highest weight $(1,\cdots,1,0,\cdots,0)$, where the integer $1$ has to be repeated as described above. For the polynomials on the supersphere this implies
\begin{eqnarray*}
\pi^\sharp_{\mS^{0|2n}}\cP\cong\bigoplus_{k=0}^{2n+1}\cH_k\cong\bigoplus_{k=0}^n L^{1|2n}_{(\underline{1}_k,\underline{0}_{n-k})}\oplus\bigoplus_{k=n+1}^{2n+1}L^{1|2n}_{(\underline{1}_{2n-k+1},\underline{0}_{k-1-n})}.
\end{eqnarray*}

The previous results also determine the representation $S(V)_k$ of supersymmetric tensors of degree $k$ for the natural module $V$ for $\mathfrak{osp}(1|2n)$. This representation can be identified with $\cP_k$. For convenience we consider the cases $k$ even and odd separately. The Fischer decomposition \eqref{superFischer} and the fact $\cH_k=\emptyset$ for $k>2n+1$ imply
\begin{eqnarray*}
S(V)_{2p}\cong \cP_{2p}&=&\bigoplus_{j=0}^{\min(n,p)}R^{2p-2j}\cH_{2j}\\
&\cong&\bigoplus_{j=0}^{\min(\lfloor \frac{n}{2}\rfloor,p)}L^{1|2n}_{(\underline{1}_{2j},\underline{0}_{n-2j})}\oplus\bigoplus_{j=\lfloor \frac{n}{2}\rfloor+1}^{\min( n,p)}L^{1|2n}_{(\underline{1}_{2n-2j+1},\underline{0}_{2j-1-n})}
\end{eqnarray*}
and
\begin{eqnarray*}
S(V)_{2p+1}\cong \cP_{2p+1}&=&\bigoplus_{j=0}^{\min(n,p)}R^{2p-2j}\cH_{2j+1}\\
&\cong&\bigoplus_{j=0}^{\min(\lfloor \frac{n-1}{2}\rfloor,p)}L^{1|2n}_{(\underline{1}_{2j+1},\underline{0}_{n-2j-1})}\oplus\bigoplus_{j=\lfloor \frac{n+1}{2}\rfloor}^{\min( n,p)}L^{1|2n}_{(\underline{1}_{2n-2j},\underline{0}_{2j-n})}.
\end{eqnarray*}

Finally, decomposition \eqref{superFischer} also leads to the following conclusion. Under the joint action of $\mathfrak{sl}_2\times\mathfrak{osp}(1|2n)$, the space $\cP=S(V)$ is isomorphic to the multiplicity free irreducible direct sum decomposition
\begin{eqnarray*}
\cP&\cong &\bigoplus_{k=0}^{n} L^{(k+1/2-n)}\otimes L^{1|2n}_{(\underline{1}_k,\underline{0}_{n-k})}\oplus \bigoplus_{k=n+1}^{2n+1} L^{(k+1/2-n)}\otimes L^{1|2n}_{(\underline{1}_{2n-k+1},\underline{0}_{k-1-n})}.
\end{eqnarray*}

\subsection{The case $m-2n\not\in-2\mN$ with $m>1$}

The simple root system of $\mathfrak{osp}(m|2n)$ as chosen in \cite{MR2395482} is used. This differs from the standard simple root system (see e.g. \cite{MR1773773}) except for $\mathfrak{osp}(2|2n)$, but is more logical for the type of representations we will study. The positive odd roots are $\epsilon_j+\delta_i$, $\delta_i$ (in case $m$ is odd) and $\epsilon_j-\delta_i$ (instead of $\delta_i-\epsilon_j$ as in \cite{MR1773773, MR051963}). The connection with the standard root system will be made in remark \ref{stanroot}.

The highest weight vector of $\cH_k$ for $\mathfrak{osp}(m|2n)$ is therefore the highest weight vector of $\cH_k^b$ for $\mathfrak{o}(m)$, which has weight $(k,0,\cdots,0)$, where $0$ is repeated $\lfloor m/2\rfloor -1$ times. This leads to the highest weight $(k,0,\cdots,0)$ for $\cH_k$ as an $\mathfrak{osp}(m|2n)$-representation, where $0$ is repeated $\lfloor m/2\rfloor +n-1$ times. Theorem \ref{irrH} therefore implies
\begin{eqnarray*}
\cH_k&\cong&L^{m|2n}_{(k,0,\cdots,0)}.
\end{eqnarray*}
A similar reasoning as in the previous section yields
\begin{eqnarray*}
\pi^\sharp_{\mS^{m-1|2n}}\cP\cong\bigoplus_{k=0}^{\infty}\cH_k\cong\bigoplus_{k=0}^\infty L^{m|2n}_{(k,0,\cdots,0)}
\end{eqnarray*}
and 
\begin{eqnarray*}
S(V)_{k}\cong \cP_k&=&\bigoplus_{j=0}^{\lfloor\frac{k}{2}\rfloor}R^{2j}\cH_{k-2j}\cong\bigoplus_{j=0}^{\lfloor k/2\rfloor}L^{m|2n}_{k-2j,0,\cdots,0}.
\end{eqnarray*}
Finally, under the joint action of $\mathfrak{sl}_2\times\mathfrak{osp}(m|2n)$, the space $\cP=S(V)$ is isomorphic to the multiplicity free irreducible direct sum decomposition
\begin{eqnarray*}
\cP&\cong &\bigoplus_{k=0}^\infty L^{(k+M/2)}\otimes L^{m|2n}_{(k,0,\cdots,0)}.
\end{eqnarray*}

\subsection{The case $m-2n\in-2\mN$}

Contrary to finite dimensional Lie algebras, finite dimensional Lie superalgebras do not possess the complete reducibility property. The representations $\cH_k$ for $M\in-2\mN$ and $2-\frac{M}{2}\le k\le 2-M$ are example of representations of $\mathfrak{osp}(m|2n)$ which is not completely reducible.
\begin{corollary}
\label{notdecomp}
If $M=m-2n\in-2\mN$ and $2-\frac{M}{2}\le k\le 2-M$, the space $\cH_k$ is not a completely reducible $\mathfrak{osp}(m|2n)$-module.
\end{corollary}
\begin{proof}
Since the inequality $2-M-k\le -\frac{M}{2}<2-\frac{M}{2}$ holds, Theorem \ref{irrH} implies that $R^{2k+M-2}\cH_{2-M-k}$ is an irreducible $\mathfrak{osp}(m|2n)$-module. The space $R^{2k+M-2}\cH_{2-M-k}$ corresponds to all the subspaces of $\cH_k$ of the form 
\begin{eqnarray*}
f_{l,k-2l-j,j}\cH_{k-2l-j}^b\otimes\cH_j^f
\end{eqnarray*}
with $l\ge k+\frac{M}{2}-1$. This follows from equation \eqref{submodule} and the unicity of $f_{l,p,q}$ in Lemma \ref{polythm} which implies 
\begin{eqnarray*}
R^{2k+M-2}f_{i,2-k-M-2i-j,j}\sim f_{k+\frac{M}{2}-1+i,2-k-M-2i-j,j}.
\end{eqnarray*}
If $\cH_k$ were completely reducible, then
\begin{eqnarray*}
\cH_k&=&R^{2k+M-2}\cH_{2-M-k}\oplus \cH_k'
\end{eqnarray*}
should hold for $\cH_k'$ some $\mathfrak{osp}(m|2n)$-representation. The complete reducibility property of finite dimensional Lie algebras implies that $\cH_k'$ is the direct sum of irreducible $\mathfrak{o}(m)\oplus\mathfrak{sp}(2n)$-representations. Theorem \ref{decompintoirreps} therefore yields
\begin{eqnarray}
\label{Hkaccent}
\cH_k'&=&\bigoplus_{j=0}^{\min(n, k)} \bigoplus_{l=0}^{\min(n-j,\lfloor \frac{k-j}{2} \rfloor, k+\frac{M}{2}-2)} f_{l,k-2l-j,j} \cH^b_{k-2l-j} \otimes \cH^f_{j}.
\end{eqnarray}
In particular, for $H_{3-M-k}^b\in\cH^b_{3-M-k}$ and $H_1^f\in\cH_1^f$, the spherical harmonic $f_{k+\frac{M}{2}-2,3-M-k,1}H_{3-M-k}^bH_1^f$ is an element of $\cH_k'$. Lemma \ref{arrows} implies that there is such a spherical harmonic and an $X\in \cU(\mathfrak{osp}(m|2n))$ (corresponding to Arrow \textbf{(4)}), such that
\begin{eqnarray*}
Xf_{k+\frac{M}{2}-2,3-M-k,1}H_{3-M-k}^bH_1^f&\in&f_{k+\frac{M}{2}-1,2-M-k,0}\cH_{2-M-k}^b\subset R^{2k+M-2}\cH_{2-M-k}.
\end{eqnarray*}
Therefore $\cH_k'$ is not an $\mathfrak{osp}(m|2n)$ representation.
\end{proof}

In these cases we find that $\cH_k$ is not equal to the irreducible highest weight module $L_{(k,0,\cdots,0)}^{m|2n}$. This representation can however still be realized by polynomials as will be shown in the subsequent Theorem \ref{Lkrep}. The following corollary will be useful for the interpretation of the representations $L_{(k,0,\cdots,0)}^{m|2n}$.
\begin{corollary}
\label{quotirr}
If $M=m-2n\in-2\mN$ and $2-\frac{M}{2}\le k\le 2-M$, the space $R^{2k+M-2}\cH_{2-M-k}$ is the maximal $\osp$-submodule of $\cH_k$. This implies that in that case,
\begin{eqnarray}
\label{voordimLk}
R^2\cP_{k-2}\cap\cH_k&=&R^{2k+M-2}\cH_{2-M-k}
\end{eqnarray}
holds and the representation $\cH_k/(R^2\cP_{k-2}\cap\cH_k)$ is irreducible.
\end{corollary}
\begin{proof}
As mentioned in equation \eqref{submodule}, in the given case $R^{2k+M-2}\cH_{2-M-k}$ is a submodule. If there would exist a larger submodule it would include one of the $\mathfrak{o}(m)\oplus \mathfrak{sp}(2n)$-modules in decomposition \eqref{Hkaccent}. For the spaces in that decomposition arrow \textbf{(1)} in Lemma \ref{arrows} always exists. Such a submodule would therefore include $\cH_k^b$ (by using arrows \textbf{(1)} and \textbf{(2)} consecutively). From the proof of theorem \ref{irrH} it is clear that a submodule of $\cH_k$ containing $\cH_k^b$ is always equal to $\cH_k$ for any values of $(m,n,k)$. 
\end{proof}

In case $m-2n\in-2\mN$ there is no Fischer decomposition as in Lemma \ref{superFischerLemma}. This also implies that the spherical harmonics $\cH_k$ do not necessarily correspond to the space $\pi^\sharp_{\mS^{m-1|2n}}(\cP_k)$. The spaces which do correspond to the polynomials on the supersphere are given by $\cP_k/(R^2\cP_{k-2})$. These spaces also form $\osp$-modules, since $\osp$ commutes with $R^2$.

These spaces $\cP_k/(R^2\cP_{k-2})$ differ from $\cH_k$ when $2-\frac{M}{2}\le k\le 2-M$ since then $R^{2k+M-2}\cH_{2-M-k}\subset \cH_k$, see equation \eqref{submodule}, and therefore $R^2\cP_{k-2}\cap \cH_k\not=\emptyset$. When $k>2-M $ or $k< 2-\frac{M}{2}$ the spaces do satisfy $\cP_k/(R^2\cP_{k-2})\cong \cH_k$ as $\osp$-modules since then $\cH_k\cap R^2\cP_{k-2}=\emptyset$ (Lemma 5.6 in \cite{Mehler} or a direct consequence of the irreducibility of $\cH_k$) and $\dim\cH_k=\dim\left(\cP_{k}/(R^2\cP_{k-2})\right)$ see formula \eqref{dimHk}. 

For $2-\frac{M}{2}\le k\le 2-M$ also the spaces $\cP_k/(R^2\cP_{k-2})$ are not irreducible since $\cH_k/(R^2\cP_{k-2}\cap\cH_k)$ is a submodule, with
\begin{eqnarray*}
\dim \cH_k/(R^2\cP_{k-2}\cap\cH_k) &<& \dim \cH_k=\dim \cP_k/(R^2\cP_{k-2}).
\end{eqnarray*}

Similarly as in the approach to $\cH_k$ in Theorem \ref{decompintoirreps} we can decompose $\cP_k/(R^2\cP_{k-2})$ into irreducible representations of $\mathfrak{o}(m)\oplus\mathfrak{sp}(2n)$ as 
\begin{eqnarray*}
\cP_k/(R^2\cP_{k-2})&=&\bigoplus_{j,l}\left(\theta^{2j}\,\cH_{k-2j-l}^b\otimes\cH^f_l+R^2\cP_{k-2}\right),
\end{eqnarray*}
which is a result of the purely bosonic and fermionic Fischer decompositions in Lemma \ref{superFischerLemma}. Again we can construct four arrows from which it will follow that the subspace $\cH_k/(R^2\cP_{k-2}\cap\cH_k)$ does not admit a compliment which is an $\osp$-representation. These results are summarized in the following lemma.
\begin{lemma}
\label{repPP}
The $\osp$-representation $\cP_k/(R^2\cP_{k-2})$ is always indecomposable. It is isomorphic to $\cH_k$ when $\cH_k$ is irreducible. When $\cH_k$ is not irreducible, also $\cP_k/(R^2\cP_{k-2})$ is not.
\end{lemma}

\subsection{The representations $L_{(k,0,\cdots,0)}^{m|2n}$}
In case $m=1$ the representation $\cH_k\cong L^{1|2n}_{(1,\cdots,1,0,\cdots,0)}$ is typical since all $\mathfrak{osp}(1|2n)$-representations are, see \cite{MR051963}. Their dimension is therefore well-known, see e.g. \cite{MR1773773}.

In case $m>1$ the representation $L_{(k,0,\cdots,0)}^{m|2n}$ is always atypical except for $m=2$ and $k=n$ or $k>2n$. This can be concluded from the atypicality conditions in chapter $36$ of \cite{MR1773773}. So the dimensions of most of the representations $L_{(k,0,\cdots,0)}^{m|2n}$ do not follow from the standard formula, they are obtained in the following theorem.

\begin{theorem}
\label{Lkrep}
For every $(m,n,k)\in\mN^3$ with $m>1$ the relation 
\begin{eqnarray*}
L_{(k,0,\cdots,0)}^{m|2n}&\cong& \cH_k/\left(\cH_k\cap R^2\cP_{k-2}\right)
\end{eqnarray*}
holds for $L_{(k,0,\cdots,0)}^{m|2n}$ the simple highest module of $\osp$ with highest weight $(k,0,\cdots,0)$ and $\cH_k$ the spherical harmonics on $\mR^{m|2n}$ of homogeneous degree $k$. 

In case $M=m-2n\in-2\mN$ and $2-\frac{M}{2}\le k\le 2-M$, the dimension of $L_{(k,0,\cdots,0)}^{m|2n}$ is given by
\begin{eqnarray*}
\dim L_{(k,0,\cdots,0)}^{m|2n}&=&\sum_{i=0}^{\min(k,2n)}\binom{2n}{i}\binom{k-i+m-1}{m-1}-\sum_{i=0}^{\min(k-2,2n)}\binom{2n}{i}\binom{k-i+m-3}{m-1}\\
&+&\sum_{i=0}^{\min(-M-k,2n)}\binom{2n}{i}\binom{2n-k-i-1}{m-1}-\sum_{i=0}^{\min(2-M-k,2n)}\binom{2n}{i}\binom{2n-k-i+1}{m-1}.
\end{eqnarray*}
In the other cases $\cH_k/\left(\cH_k\cap R^2\cP_{k-2}\right)=\cH_k$ and
\begin{eqnarray*}
\dim L_{(k,0,\cdots,0)}^{m|2n}&=&\sum_{i=0}^{\min(k,2n)}\binom{2n}{i}\binom{k-i+m-1}{m-1}-\sum_{i=0}^{\min(k-2,2n)}\binom{2n}{i}\binom{k-i+m-3}{m-1}
\end{eqnarray*}
holds.
\end{theorem}
\begin{proof}
The relation $L_{(k,0,\cdots,0)}^{m|2n}\cong \cH_k/\left(\cH_k\cap R^2\cP_{k-2}\right)$ is a direct consequence of Theorem \ref{irrH} and Corollary \ref{quotirr}. The dimensions then immediately follow from equation \eqref{voordimLk} and the dimensions of the spaces $\cH_k$ in equation \eqref{dimHk}.
\end{proof}

In \cite{MR0621253} the representations $L_{(k,0,\cdots,0)}^{m|2n}$ were constructed as the Cartan product inside tensor products of the form $(L_{(1,0,\cdots,0)}^{m|2n})^{\otimes k}$. Here, the Cartan product corresponds to the traceless supersymmetric part, which can be identified with $\cP_k/(R^2\cP_{k-2})$. However as we have seen this construction only holds when $m-2n\not\in-2\mN$, due to the non-completely reducibility, see lemma \ref{repPP}. This was overlooked in the formal approach in \cite{MR0621253} were for instance the number $m-2n$ appears as a pole in equation $(4.21)$. The correct construction of $L_{(1,0,\cdots,0)}^{\otimes k}$ inside the supersymmetric tensor products of $V=L_{(1,0,\cdots,0)}^{m|2n}$ is given in theorem \ref{Lkrep}.

\begin{remark}
\label{stanroot}
In the standard choice of positive and negative roots with the corresponding distinguished basis (see e.g. \cite{MR1773773}) would have been used the highest weight of the representation $\cH_k$ would be given by 
\[\delta_1+\cdots+\delta_{\min(k,n)}+(k-\min(k,n))\epsilon_1\]
in stead of $k\epsilon_1$. This shows it is more elegant to use the choice of positive roots made in \cite{MR2395482} for these kind of representations.
\end{remark}

In \cite{MR2395482} it was proved that for $M=m-2n>2$ the following branching rule holds:
\begin{eqnarray}
\label{branching}
L_{(k,0,\cdots,0)}^{m|2n}&\cong&\bigoplus_{l=0}^kL_{(l,0,\cdots,0)}^{m-1|2n} \qquad\mbox{as an $\mathfrak{osp}(m-1|2n)$-module.}
\end{eqnarray}
This can immediately be extended to $m-2n=2$, but not to the case $m-2n\le 1$ because of the appearance of not completely reducible representations. The method applied in \cite{MR2395482} does however lead to the relation
\begin{eqnarray*}
\cP_k/(R^2\cP_{k-2})&\cong&\bigoplus_{l=0}^k\left(\cP'_l/(R_1^2\cP'_{l-2})\right)\qquad\mbox{as an $\mathfrak{osp}(m-1|2n)$-module, for general $M$,}
\end{eqnarray*}
where $\cP'$ denotes the polynomials on $\mR^{m-1|2n}$ and $R_1^2$ the generalized norm squared on $\mR^{m-1|2n}$. This leads to the following two conclusions based on Theorem \ref{Lkrep}, Lemma \ref{repPP} and Corollary \ref{notdecomp}:
\begin{itemize}
\item If $m-2n\le 1$ but $m-2n\not\in-2\mN$, $L_{(k,0,\cdots,0)}^{m|2n}$ is not completely reducible as an $\mathfrak{osp}(m-1|2n)$-representation if $k\ge 2+\frac{1-M}{2}$. Relation \eqref{branching} still holds if $k< 2+\frac{1-M}{2}$.
\item If $m-2n\in-2\mN$, the space $\cP_k/(R^2\cP_{k-2})$ which is not necessarily completely reducible as an $\osp$-representation decomposes into irreducible $\mathfrak{osp}(m-1|2n)$-representations as 
\begin{eqnarray}
\label{PLk}
\cP_k/(R^2\cP_{k-2})&\cong& \bigoplus_{l=0}^kL_{(l,0,\cdots,0)}^{m-1|2n}.
\end{eqnarray}
Therefore, relation \eqref{branching} still holds if $k<2-\frac{M}{2}$ or $k>2-M$.
\end{itemize}

We summarize these results in the following theorem and calculate the branching rule for the other cases when $m-2n\in-2\mN$. Therefore we obtained all the branching rules for when $L_{(k,0,\cdots,0)}^{m|2n}$ is a completely reducible $\mathfrak{osp}(m-1|2n)$-representation.

\begin{theorem}
\label{branchingThm}
In case $m-2n>1$, or $m-2n\in 1-2\mN$ with $k<2+\frac{1-m}{2}+n$, or $m-2n\in-2\mN$ with $k<2-\frac{m}{2}+n$ or $k>2-m+2n$, the branching rule
\begin{eqnarray*}
L_{(k,0,\cdots,0)}^{m|2n}&\cong&\bigoplus_{l=0}^kL_{(l,0,\cdots,0)}^{m-1|2n} \qquad\mbox{as an $\mathfrak{osp}(m-1|2n)$-module},
\end{eqnarray*}
holds. In case $m-2n\in-2\mN$ and $2-\frac{m}{2}+n\le k\le 2-m+2n$, the branching rule
\begin{eqnarray*}
L_{(k,0,\cdots,0)}^{m|2n}&\cong&\bigoplus_{l=3-m+2n-k}^kL_{(l,0,\cdots,0)}^{m-1|2n} \qquad\mbox{as an $\mathfrak{osp}(m-1|2n)$-module},
\end{eqnarray*}
holds. In the other cases ($m-2n\in1-2\mN$ with $k\ge 2+\frac{1-m}{2}+n$), $L_{(k,0,\cdots,0)}^{m|2n}$ is not completely reducible as an $\mathfrak{osp}(m-1|2n)$-representation.
\end{theorem}

\begin{proof}
The first cases were already proved. Theorem \ref{Lkrep} and equation \eqref{PLk} imply that in the second case
\begin{eqnarray*}
L_{(k,0,\cdots,0)}^{m|2n}\cong \cH_k/\left(\cH_k\cap R^2\cP_{k-2}\right) \subset \cP_k/(R^2\cP_{k-2}) \cong \bigoplus_{l=0}^kL_{(l,0,\cdots,0)}^{m-1|2n}.
\end{eqnarray*}
This implies that 
\begin{eqnarray}
\label{knownformbranch}
L_{(k,0,\cdots,0)}^{m|2n}&\cong&\bigoplus_{p\in I}L_{(p,0,\cdots,0)}^{m-1|2n} \qquad\mbox{as an $\mathfrak{osp}(m-1|2n)$-module,}
\end{eqnarray}
with $I\subset\{0,\cdots,k\}$. We look at the decomposition of $L_{(k,0,\cdots,0)}^{m|2n}$ into simple $\mathfrak{o}(m)\oplus \mathfrak{sp}(2n)$-modules. Theorem \ref{Lkrep} and the proof of Corollary \ref{notdecomp} (in particular equation \eqref{Hkaccent}) imply that $\mathfrak{sp}(2n)$-trivial the part of this decomposition is given by
\begin{eqnarray*}
L_{(k,0,\cdots,0)}^{m|2n}&\rightarrow&\bigoplus_{l=0}^{k+\frac{M}{2}-2}L_{(k-2l,0,\cdots,0)}^{m|0},
\end{eqnarray*}
since $\min(n,\lfloor \frac{k}{2}\rfloor,k+\frac{M}{2}-2)=k+\frac{M}{2}-2$. Branched to $\mathfrak{o}(m-1)$ this gives $\bigoplus_{l=0}^{k+\frac{M}{2}-2}\bigoplus_{j=0}^{k-2l}L_{(k-2l-j,0,\cdots,0)}^{m-1|0}$. Theorem \ref{decompintoirreps} and Theorem \ref{Lkrep} imply the $\mathfrak{sp}(2n)$-trivial part of $L_{(p,0,\cdots,0)}^{m-1|2n}$ for $p\le k$ is
\begin{eqnarray*}
L_{(p,0,\cdots,0)}^{m-1|2n}&\to& \bigoplus_{l=0}^{\lfloor \frac{p}{2}\rfloor}L_{(p-2l,0,\cdots,0)}^{m-1|0}
\end{eqnarray*}
since $\min(n,\lfloor \frac{p}{2}\rfloor)=\lfloor \frac{p}{2}\rfloor$ for $p\le k$. If equation \eqref{knownformbranch} holds,  the equation
\begin{eqnarray}
\label{eqforbranching}
\bigoplus_{l=0}^{k+\frac{M}{2}-2}\bigoplus_{j=0}^{k-2l}L_{(k-2l-j,0,\cdots,0)}^{m-1|0}&=&\bigoplus_{p\in I}\bigoplus_{l=0}^{\lfloor \frac{p}{2}\rfloor}L_{(p-2l,0,\cdots,0)}^{m-1|0}
\end{eqnarray}
must hold as well. It turns out that this equation is enough to determine $I$. We introduce the shorthand notation $(q)=L_{(q,0,\cdots,0)}^{m-1|0}$. In the left-hand side of equation \eqref{eqforbranching} the module $(j)$ appears $1+\lfloor\frac{k-j}{2}\rfloor$ times for $j\ge -k-M+3$. To obtain all these in the right-hand side of equation \eqref{eqforbranching}, $\{3-M-k,\cdots,k\}\subset I$ must hold. Since the module $(j)$ appears $k+M/2-1$ times for $j\le -k-M+2$ we find  $\{3-M-k,\cdots,k\}= I$, because otherwise there would be too many $\mathfrak{o}(m-1)$-modules in the right-hand side.
\end{proof}

In the statements above we assumed $m\not=2$ and $m\not=1$. The equivalent statements for those dimensions are straightforward.

\subsection{The algebra $\mathfrak{osp}(4n+1|2m)$}

In the theory of Howe dual pairs the dual algebras are each others' centralizers inside a bigger algebra, such that the relevant representation (which has a multiplicity free decomposition into irreducible pieces under the action of the dual pair) constitutes an irreducible representation of the big algebra.

For the Fischer decomposition of polynomials on $\mR^m$ this is $\mathfrak{o}(m)\times \mathfrak{sl}_2\subset \mathfrak{sp}(2m)$, since Howe \cite{MR0986027} originally considered dual pairs inside the symplectic group. The polynomials $\mR[x_1,\cdots, x_m]$ fall apart into two irreducible $\mathfrak{sp}(2m)$-representations, corresponding to the even and odd polynomials. Therefore there are two non-isomorphic representations corresponding to the Howe dual pair $\mathfrak{sp}(2m)\supset \mathfrak{o}(m)\times \mathfrak{sl}_2$. To obtain the polynomials as one irreducible representation, $\mathfrak{sp}(2m)$ needs to be embedded inside the superalgebra $\mathfrak{osp}(1|2m)$.

In superspace this construction of the bigger algebra such that $\cP$ is an irreducible representation is as follows. The Lie superalgebra spanned by the differential operators $X_i$, $\partial_{X_i}$, $X_iX_j$, $\partial_{X_i}\partial_{X_j}$ and $X_i\partial_{X_j}+\frac{(-1)^{[i]}\delta_{ij}}{2}$ where the generators $X_i$ and $\partial_{X_i}$ have gradation $1-[i]$, is $\mathfrak{osp}(4n+1|2m)$. This is closely related to the oscillator realization of the orthosymplectic Lie superalgebra, see chapter 29 in \cite{MR1773773}, or \cite{MR1272070}. Clearly the polynomials $\cP$ form an irreducible representation with action of $\mathfrak{osp}(4n+1|2m)$ given above. This corresponds to the irreducible highest weight representation
\begin{eqnarray*}
L^{4n+1|2m}_{(\frac{1}{2},\cdots, \frac{1}{2},-\frac{1}{2},\cdots,-\frac{1}{2})},
\end{eqnarray*}
where $1/2$ is repeated $2n$ times and $-1/2$ $m$ times, the highest weight vector is the constant $1\in\cP$.

Also the realizations of $\osp$ and $\mathfrak{sl}_2$ as differential operators on $\cP$ are clearly embedded in this realization of $\mathfrak{osp}(4n+1|2m)$. So we have obtained the bigger algebra in which $\osp$ and $\mathfrak{sl}_2$ are each others' centralizers:
\begin{eqnarray*}
\osp\times \mathfrak{sl}_2&\subset& \mathfrak{osp}(4n+1|2m).
\end{eqnarray*}
If we would leave out the differential operators $X_i$ and $\partial_{X_i}$ we would obtain the superalgebra $\mathfrak{osp}(4n|2m)$. For this algebra, the representation $\cP$ decomposes into two irreducible representations corresponding to the polynomials of even and odd degree, as in the classical case. So we obtain
\begin{eqnarray*}
\mathfrak{osp}(4n|2m)\supset\mathfrak{osp}(m|2n)\times \mathfrak{sl}_2&\mbox{ as a generalization of}& \mathfrak{sp}(2m)\supset\mathfrak{o}(m)\times \mathfrak{sl}_2.
\end{eqnarray*}

\subsection*{Acknowledgment}
The author would like to thank prof. Ruibin Zhang, prof. Joris Van der Jeugt and prof. Bent \O rsted for helpful suggestions and comments.


\end{document}